\documentclass{article}

\usepackage[ruled,algo2e]{algorithm2e}
\usepackage{microtype}
\usepackage{graphicx}
\usepackage{subfigure}
\usepackage{booktabs} 

\usepackage{hyperref}

\usepackage{tabularx}
\usepackage{makecell}
\usepackage{booktabs}
\usepackage{multirow}

\usepackage[accepted]{icml2024}

\usepackage{amsmath}
\usepackage{amssymb}
\usepackage{mathtools}
\usepackage{amsthm}

\usepackage[capitalize,noabbrev]{cleveref}

\theoremstyle{plain}
\newtheorem{theorem}{Theorem}[section]

\theoremstyle{definition}

\theoremstyle{remark}

\usepackage[textsize=tiny]{todonotes}

\newcommand{\Rd}{\mathbb{R}^d}
\newcommand{\R}{\mathbb{R}}

\DeclareMathOperator*{\argmin}{arg\,min}

\icmltitlerunning{Computing Transition Pathways Using Deep RL}

\begin{document}
\twocolumn[
\icmltitle{Computing Transition Pathways for the Study of Rare Events \\ Using Deep Reinforcement Learning}




\begin{icmlauthorlist}
\icmlauthor{Bo Lin}{yyy}
\icmlauthor{Yangzheng Zhong}{yyy}
\icmlauthor{Weiqing Ren}{yyy}
\end{icmlauthorlist}

\icmlaffiliation{yyy}{Department of Mathematics, National University of Singapore, 10 Lower Kent Ridge Road 119076, Singapore}

\icmlcorrespondingauthor{Bo Lin}{matboln@nus.edu.sg}
\icmlcorrespondingauthor{Weiqing Ren}{matrw@nus.edu.sg}


\vskip 0.3in
]



\printAffiliationsAndNotice{}  

\begin{abstract}
Understanding the transition events between metastable states in complex systems is an important subject in the fields of computational physics, chemistry and biology. The transition pathway plays an important role in characterizing the mechanism underlying the transition, for example, in the study of conformational changes of bio-molecules. In fact, computing the transition pathway is a challenging task for complex and high-dimensional systems. 
In this work, we formulate the path-finding task as a cost minimization problem over a particular path space. The cost function is adapted from the Freidlin-Wentzell action functional so that it is able to deal with rough potential landscapes. 
The path-finding problem is then solved using a actor-critic method based on the deep deterministic policy gradient algorithm (DDPG).
The method incorporates the potential force of the system in the policy for generating episodes 
and combines physical properties of the system with the learning process for molecular systems. 
The exploitation and exploration nature of reinforcement learning enables the method to efficiently sample the transition events and compute the globally optimal transition pathway.
We illustrate the effectiveness of the proposed method using three benchmark systems including an extended Mueller system and the Lennard-Jones system of seven particles.
\end{abstract}

\section{Introduction}
Understanding the transition events of dynamical systems is a fundamental but challenging task in many science problems, including chemical reaction, conformational changes of bio-molecules and nucleation events during phase transitions. 
For the system, the transition occurs by crossing a typical energy barrier separating two metastable states in the presence of random perturbations. 
The general disparity between the effective thermal energy and the energy barrier usually leads to a long-waiting period for the system around metastable states before a sudden transition from one state to another emerges. In this setting, we call the transition as a rare event. 
The major interest for the study of rare events is to compute the mechanism underlying the transition events, such as the transition pathway and transition rate. 
In the past, a number of research works have been devoted to developing efficient approaches for investigating the transition mechanism. Well-known methods include the transition path sampling~\cite{bolhuis2002transition,dellago2002transition}, the string method~\cite{weinan2002string,weinan2005finite,ren2005transition,maragliano2006string,ren2007simplified}, the action-based method~\cite{olender1996calculation} and the accelerated molecular dynamics~\cite{voter1997hyperdynamics}. Also, the recent works in Ref.~\cite{khoo2019solving,li2019computing,rotskoff2022active} proposed deep learning based methods for computing the committor function, which is a central object in the transition path theory for understanding the transition mechanism~\cite{ren2005transition,vanden2010transition}. 

The transition pathway in the zero-temperature limit is characterized by the minimum energy path (MEP). The MEP is a path defined in the configuration space along which the tangent of the path is parallel to the potential force. Classical methods for identifying the MEP include the nudged elastic band method~\cite{Newell81}, the string method~\cite{weinan2002string,weinan2005finite,ren2005transition,maragliano2006string,ren2007simplified} and the conjugate peak refinement method~\cite{fischer1992conjugate}. 
The Freidlin-Wentzell theory of large deviations provides a variational characterization of the most likelihood path via a path-based action functional. Based on this characterization, Ref.~\cite{weinan2004minimum,zhou2008adaptive,heymann2008geometric} developed the minimum action methods for computing the minimum action path by minimizing the functional over a path space where the two ends of path are constrained to particular states. Also, the Onsager-Machlup functional, first introduced by Onsager and Machlup~\cite{onsager1953fluctuations}, was used to compute the most probable path associated with a finite-time horizon for systems at finite noise in many applications~\cite{wang2010kinetic,fujisaki2010onsager,du2021graph}. From another perspective, one can recast the variational formulation of the transition path as a finite-horizon optimal control problem~\cite{fleming1977exit,grafke2019numerical}. The control problem was recently solved by a reinforcement learning algorithm~\cite{guo2024deep}.

The situation is quite different when the potential energy surface is rough, which is the general case for practical molecular systems. In this case, the ensemble of transition paths can be characterized by a tube in the configuration space connecting the metastable states inside which the transition occurs with high probability~\cite{ren2005transition}. The rough potential energy surface contains numerous saddle points, most of which are separated by energy barriers comparable to or less than the thermal energy and thus do not act as bottlenecks of the transition. 
As a consequence, the potential force should not be straightforwardly used in the variational characterization of the transition pathway.

In this paper, we propose a deep reinforcement learning method for computing the transition pathway of the system with rough potential landscapes. We formulate the path-finding task as a cost minimization problem over a particular path space. 
To tackle the difficulties arising from the roughness of the potential landscape, we adapt the Freidlin-Wentzell functional and propose a cost function involving an effective force function. In the zero-temperature limit, the effective force simply reduces to the potential force of the system. The formulated path-finding problem is over a constrained sequence of states in the configuration space, 
where the optimal times slices are determined using numerical quadratures. 

In recent years, the advances in reinforcement learning algorithms have made successful applications in many sequential decision making problems, including video games, robotic control and autonomous driving. 
A general class of reinforcement learning algorithms are based on computation of the state-action value function, such as the Deep Q Network (DQN) algorithm~\cite{mnih2013playing,mnih2015human}, 
the deterministic policy gradient (DPG) algorithm~\cite{silver2014deterministic} and the deep DPG (DDPG) algorithm~\cite{lillicrap2015continuous}. In particular, the DQN algorithm utilizes the deep neural networks as the function approximators, where target networks and replay buffer are introduced to stabilize the algorithm. Furthermore, the DDPG algorithm combined the DQN with the actor-critic approach based on the policy gradient to adapt the algorithm to a broader case with continuous and high-dimensional action space. 

In this work, we solve the formulated cost minimization problem for identifying the transition pathway using the actor-critic method based on the DDPG algorithm. The critic and actor functions are parameterized by neural networks. To target the exploration in the region of interest, the method employs a stochastic policy based on the actor function with random noise and the potential force of the system for generating episodes. Also, we utilize target networks and a replay buffer to address the possible instability issue in the learning process. For molecular systems, to enhance the learning efficiency, we incorporate physical properties of the system into the critic and actor networks. 
The exploitation and exploration nature of reinforcement learning together with these techniques establish a stable and efficient algorithm for sampling transition events and computing the globally optimal transition pathway for high-dimensional systems with rough potential landscapes. 
We demonstrate the effectiveness of the method using three numerical examples including a two-dimensional system, an extended Mueller system and the Lennard-Jones system of seven particles.

The paper is organized as follows. In Section~\ref{Method}, we introduce the background and the formulated path-finding problem and propose the reinforcement learning algorithm. The numerical examples are presented in Section~\ref{Examples}. In Section~\ref{Conclusion}, we draw the conclusions.

\section{Method}\label{Method}
We consider a dynamical system in the configuration space $\Rd$, which is modelled by the over-damped Langevin equation:
\begin{equation}\label{SDE}
    d x_t = -\nabla V(x_t) dt + \sqrt{2\epsilon} dW_t,\quad t>0
\end{equation}
where $V(x)$ is a potential function, $W_t$ is a standard Brownian motion and $\epsilon$ is a parameter specifying the strength of the noise which is called the temperature of the system. 
The equilibrium distribution of the system is known as the Boltzmann density function $\rho(x)=Z^{-1}\exp(-\frac{1}{\epsilon}V(x))$, where $Z$ is a normalization constant. 
Consider the general situation where there are two metastable states $A$ and $B$ for the system. 
A transition pathway between the two metastable states $A$ and $B$ is defined as a curve in configuration space connecting the two states.

For a time $T>0$, we denote by $\mathbb{C}_{[0,T]}$ the set of all absolute continuous functions in the configuration space over the time interval $[0,T]$ connecting the two metastable states $A$ and $B$. 
The Freidlin-Wentzell action functional for a given path $\varphi\in\mathbb{C}_{[0,T]}$ is defined as
\begin{equation}\label{FW_action}
    S_T[\varphi] = \int_0^T \frac{1}{4} \left| \varphi'(t)+ \nabla V(\varphi(t))\right|^2 dt.
\end{equation}
According to the Freidlin-Wentzell theory of large deviations~\cite{freidlin2012random}, when the noise $\epsilon$ is sufficiently small, for small number $\delta>0$, the probability that the system~\eqref{SDE} stays in the neighborhood of a given path $\varphi\in\mathbb{C}_{[0,T]}$ over the time interval $[0,T]$ can be estimated by
\begin{equation*}
    \mathbb{P}\left[\max_{0\leqslant t\leqslant T}\|x(t)-\varphi(t)\|<\delta\right]\approx \exp\left(-\frac{S_T[\varphi]}{\epsilon}\right).
\end{equation*}
Therefore, in the zero-temperature limit, the most probable transition path of the system within the time horizon $[0,T]$ can be characterized by a minimizer of the functional~\eqref{FW_action}:
\begin{equation}\label{min_action}
    \varphi^*_T = \argmin_{\varphi\in \mathbb{C}_{[0,T]}} S_T[\varphi].
\end{equation}

The path $\varphi^*_T$ is also referred to as the minimum action path. In the infinite-time limit by letting $T$ goes to infinity, the action functional of the minimum action path, $S_T[\varphi^*_T]$, converges to the infimum value
\begin{equation}\label{min_action_T}
    \inf_{T>0} \inf_{\varphi\in\mathbb{C}_{[0,T]}} S_T[\varphi],
\end{equation}
where the infimum is over all times $T$ and the corresponding path space $\mathbb{C}_{[0,T]}$. 
Furthermore, the graph limit of the minimum action path is a minimum energy path (MEP) of the system. The MEP is, by definition, a curve in the configuration space along which the tangent of the curve is parallel to the potential force $-\nabla V(x)$.

In this work, we aim to find the transition pathway with minimal action in a path space that is close to $\cup_T \mathbb{C}_{[0,T]}$. 
To deal with rough potential landscapes, we propose a cost function by adapting the Freidlin-Wentzell functional. Then we propose a deep reinforcement learning algorithm to solve the path-finding problem.

\subsection{Formulation of the Problem}

For a small number $\gamma>0$, we consider a path space $\mathbb{C}_{\gamma}$ consisting of continuous piecewise linear functions, whose graph is a polygonal chain connecting the two states $A$ and $B$. Each path $\varphi(t)$ in the space $\mathbb{C}_{\gamma}$ is represented by a sequence of states $(z_0,\dots,z_N)$ in the configuration space, together with the corresponding time points $(t_0,\dots,t_{N})$ (See Fig.~\ref{fig0}). The sequence $(z_0,\dots,z_N)$ forms a polygonal chain connecting $A$ and $B$ where the line segments are of equal length $\gamma$ except the last one, {\it i.e.}
\begin{equation}\label{constraint}
\begin{aligned}
    &z_0= A;\quad z_N= B;\\
    &\gamma = \left| z_i-z_{i+1}\right|,\ 0\leqslant i\leqslant N-2;\\
    &\gamma \geqslant \left| z_{N-1}-z_{N}\right|.
\end{aligned}
\end{equation} 

\begin{figure}[!h]
\centering
\includegraphics[width=\linewidth]{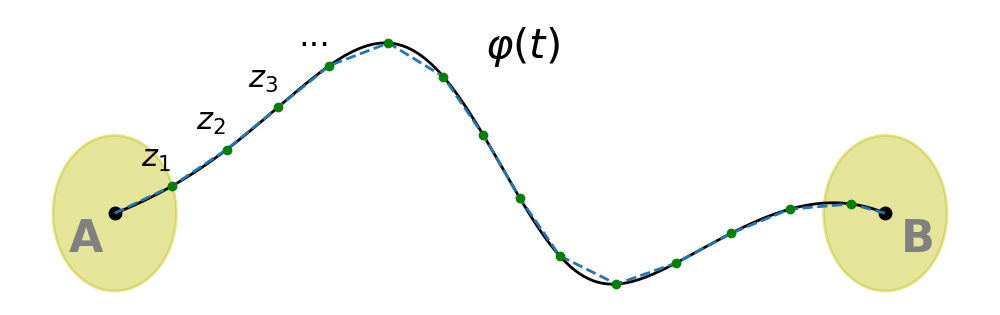}
\caption{Schematic illustration of an absolute continuous path (black solid line) connecting the two metastable states $A$ and $B$, which is approximated by a path $\varphi(t)$ (green dashed line) in $\mathbb{C}_{\gamma}$ represented by a polygonal chain with a sequence of states $(z_0,\dots,z_N)$.}
\label{fig0}
\end{figure}

Over the time interval $[t_i,t_{i+1}]$, the path $\varphi(t)$ is a straight line connecting $z_i$ and $z_{i+1}$ with uniform derivatives,
\begin{equation*}
    \varphi(t)=z_i + \dfrac{z_{i+1}-z_{i}}{h_i}(t-t_i),\quad t_i\leqslant t\leqslant t_{i+1}
\end{equation*}
where the time slice $h_i=t_{i+1}-t_i$. We denote the function $\varphi^i(t)=\varphi(t+t_i)$ for $t\in[0,h_i]$. 
One can show that the set $\cup_{\gamma>0}\mathbb{C}_{\gamma}$ is a dense subset of $\cup_{T>0}\mathbb{C}_{[0,T]}$ (See Appendix~\ref{function_space}). Next, we consider the following action minimization problem restricted to the path space $\mathbb{C}_{\gamma}$ for computing the transition pathway of the system,

\begin{equation}\label{min_Sr}
\begin{aligned}
    &\inf_{\varphi\in\mathbb{C}_{\gamma}}
\sum_{i=0}^{N-1} \tilde{L}(z_i,z_{i+1};h_i),\quad \text{where }
\\
&\tilde{L}(z_i,z_{i+1};h_i)= \int_{0}^{h_i} \left| \frac{z_{i+1}-z_i}{h_i}+\nabla V(\varphi^i(t))\right|^2 dt.\qquad
\end{aligned}
\end{equation}
The optimization problem is, in principle, over a finite set of states $(z_0,\dots,z_N)$ in the configuration space subject to the constraints~\eqref{constraint}, together with the time slices $(h_0,\dots,h_{N-1})$.

{\bf Optimal Time Slices.} In the problem~\eqref{min_Sr}, for a particular sequence of states $(z_0,\dots,z_N)$, each optimal time slice $h_i^*$ is a minimizer to the individual integral $\tilde{L}(z_i,z_{i+1};h_i)$. However, it is not straightforward to obtain an analytical solution for $h_i^*$ in general as the integral involves the potential force. Instead, one can compute an optimal solution $h^{*}_i$ by approximating the integral $\tilde{L}(z_i,z_{i+1};h_i)$ using a mid-point numerical quadrature, 
\begin{align}
    \min_{h_i}\tilde{L}(z_i,&z_{i+1};h_i) \approx \min_{h_i}h_i \left| \frac{z_{i+1}-z_i}{h_i}+\nabla V(z_{i+1/2})\right|^2 \nonumber\\
    &= 2 |z_{i+1}-z_i|\cdot
    |\nabla V(z_{i+1/2})| \nonumber\\
    &\qquad +2 \langle z_{i+1}-z_i ,  \nabla V(z_{i+1/2})\rangle.\label{cost1}
\end{align}

where the mid-point $z_{i+1/2}=(z_i+z_{i+1})/2$ and the minimum value for $\tilde{L}$ is achieved at $h_i^*=|z_{i+1}-z_{i}|\left/|\nabla V(z_{i+1/2})|\right.$. We denote the minimum value in Eq.~\eqref{cost1} by $R(z_i,z_{i+1})$ and refer to it as the cost between $z_i$ and $z_{i+1}$. In principle, the integral $\tilde{L}(z_i,z_{i+1};h_i)$ in Eq.~\eqref{min_Sr} can be approximated using any suitable numerical quadratures (See Appendix~\ref{quadrature}).

Therefore, the task of finding the optimal transition  path in the space $\mathbb{C}_{\gamma}$ as in Eq.~\eqref{min_Sr} can be formulated as a cost minimization problem over the states $(z_0,\dots,z_N)$:
\begin{equation}\label{prob1}
\argmin_{(z_0,\dots,z_N)} \sum_{i=0}^{N-1} R(z_i,z_{i+1}),
\end{equation}
where the states $(z_0,\dots,z_N)$ representing a transition path $\varphi(t)$ are subject to the constraints in Eq.~\eqref{constraint}.

{\bf Effective Force.} The situation is quite different when the potential energy surface is rough, which is the typical case for practical molecular systems. In this case, the rough potential energy surface may contain numerous saddle points, most of which do not act as bottleneck of the transition as potential barriers separating those saddle points are less than or comparable to the thermal energy~\cite{ren2005transition}. As a consequence, the Freidlin-Wentzell path functional as in Eq.~\eqref{FW_action} directly involving the potential force is no longer valid for quantifying the transition events.

For a particular line segment connecting $z_i$ and $z_{i+1}$ in the path, consider the original dynamics~\eqref{SDE} starting from the mid-point $z_{i+1/2}$:
\begin{equation}\label{SDE0}
    dx_t = -\nabla V(x_t) dt + \sqrt{2\epsilon}dW_t,\quad x_0=z_{i+1/2}.
\end{equation}
For $h>0$, integrating both sides of the equation over the time interval $[0,h]$ gives 
\begin{equation*}
 x_h = z_{i+1/2} - \int_{0}^{h} \nabla V(x_t) dt + \xi,\quad \xi\sim \mathcal{N}(0,2\epsilon h).
\end{equation*}
We treat the potential force in the integral as a uniform value around the state $z_{i+1/2}$ and refer to it as the effective force at $z_{i+1/2}$. The formal definition of the effective force at $z_{i+1/2}$ is given by
\begin{equation}\label{def_F}
    F_{\epsilon}(z_{i+1/2}) =
    \mathbb{E}_{\epsilon}\left[\frac{x_h-x_0}{h}: x_0=z_{i+1/2}\right],
\end{equation}
where the expectation is over a state ensemble of the system at time $h$ starting from the state $z_{i+1/2}$ at the temperature $\epsilon$. In practice, we approximate the effective force $F_{\epsilon}(z_{i+1/2})$ using $M$ short trajectories following the dynamics~\eqref{SDE0} over the time interval $[0,h]$ where the state at time $h$ for each trajectory is denoted by $x_j^e$,
\begin{equation}\label{MC}
    F_{\epsilon}(z_{i+1/2}) \approx
    \frac{1}{M}\sum_{j=1}^M \frac{x_j^e-z_{i+1/2}}{h}.
\end{equation}
We then define the cost function
\begin{equation}\label{cost2}
\begin{aligned}
    R_{\epsilon}(z_i,z_{i+1}) &= \min_{h_i} h_i\left| \frac{z_{i+1}-z_{i}}{h_i}- F_{\epsilon}(z_{i+1/2}) \right|^2 \\
    &= 2 |z_{i+1}-z_{i}| \cdot\left|F_{\epsilon}(z_{i+1/2})\right| \\
    &\qquad - 2 (z_{i+1}-z_{i})\cdot F_{\epsilon}(z_{i+1/2}).
\end{aligned}
\end{equation}
Note that the cost function reduces to the Freidlin-Wentzell cost in Eq.~\eqref{cost1} when $\epsilon$ and $h$ tends to zero since the expectation~\eqref{def_F} asymptotically converges to the potential force $-\nabla V(z_{i+1/2})$. For simplicity, we denote by $R_0$ the zero-temperature cost function $R$ as defined in Eq.~\eqref{cost1}. For the system with rough potential landscapes at temperature $\epsilon$, the transition pathway connecting $A$ and $B$ can be computed by solving the cost minimization problem
\begin{equation}\label{prob2}
    \argmin_{(z_0,\dots,z_N)} \sum_{i=0}^{N-1} R_{\epsilon}(z_i,z_{i+1}),
\end{equation}
where the states $(z_0,\dots,z_N)$ in the configuration space are subject to the constraints in Eq.~\eqref{constraint}.

\subsection{Reinforcement Learning Method}\label{RL}

To solve the path-finding problem in Eq.~\eqref{prob2}, we define a Markov decision process with a state space $\mathcal{X}=\Rd$ and continuous action space $\mathcal{A}=\{x\in\Rd:|x|=1\}$. During the process, an agent interacts with the environment at discrete time steps. At each step, the agent takes an action, observes the next state and receives a running cost (or reward). Specifically, we set the transition dynamics and cost function $r(s,a)$ to be deterministic and consistent with the problem~\eqref{prob2}, {\it i.e.} the next state after taking action $a_t$ at state $s_t$ is given by $s_{t+1}=s_t+\gamma\cdot a_t$ and the received cost is defined as $r_t = R_{\epsilon}(s_{t},s_{t+1})$ as in Eq.~\eqref{cost1} for $\epsilon=0$ or Eq.~\eqref{cost2} for $\epsilon>0$. The terminal condition is that the agent arrives in the region $\Omega_{\gamma}=\{x\in\Rd:|x-B|<\gamma\}$. 

The agent's behaviour is described by a policy $\pi$, which gives the action $\pi(s)$ for each state $s$. For a given policy $\pi$, the return from a state $s_t$ is defined as the sum of future costs, $R_t=r_t + \cdots + r_T$, where $T$ denotes the terminal time of the process. Our goal is to learn a policy $\pi^*$ which minimizes the return for each state $s\in \mathcal{X}$. The state-action function $Q(s,a)$ associated with the optimal policy $\pi^*$ is defined as the return after taking action $a$ at $s$ and thereafter following the policy $\pi^*$. Many reinforcement learning algorithms for computing the function $Q(s,a)$ are based on a recursive relationship known as the Bellman equation:
\begin{equation*}
    Q(s_t,a_t) = r(s_t,a_t) + \min_{b\in \mathcal{A}} Q(s_{t+1},b).
\end{equation*}

To deal with the task with continuous action space, here we use an actor-critic approach based on the DDPG algorithm~\cite{lillicrap2015continuous}. 
The critic function is the state-action function $Q(s,a)$ which is parameterized using a neural network $\tilde{Q}_{\theta}(s,a):\mathcal{X}\times\mathcal{A}\rightarrow [0,1]$,
\begin{equation}\label{Q_form}
    Q_{\theta}(s,a) =
    \begin{cases}
        \lambda \tilde{Q}_{\theta}(s,a), & s\notin \Omega_{\gamma} \\
        R_{\epsilon}(s,B), & s\in \Omega_{\gamma}
    \end{cases}
\end{equation}
where $\lambda>0$ is predefined parameter which specifies the range of the critic function as $[0,\lambda]$. The actor function $\mu(s)=\argmin_a Q(s,a)$ corresponding to the optimal policy specifies the optimal action at each state. To construct an actor network mapping states to unit vectors in $\Rd$, we represent the actor function as a normalized vector based the cosines of a hidden actor $\tilde{\mu}_{\theta}:\Rd\rightarrow \Rd$,
\begin{equation}\label{mu_form}
    \mu_{\theta}(s) = \Theta[\cos(\tilde{\mu}_{\theta}(s))],
\end{equation}
where $\Theta[v]=v/|v|$ is the normalization function with an input vector $v\in\Rd$.
Note that the actor function is periodic over the hidden actor $\tilde{\mu}_{\theta}(s)$ with period $2\pi$.

{\bf Data Generation.} We sample data by generating episodes where the initial states are produced from a particular distribution $p(s)$ and subsequently the action of the agent is selected based on a stochastic policy. To facilitate the exploration of possible transition pathways, we add noise to the hidden actor function $\tilde{\mu}_{\theta}$ in the selection of action. Meanwhile, we aim to target the exploration in the region of interest which excludes states in the configuration space with low equilibrium densities $\rho(x)$. Thus, we take the action conducted at step $t$ as  $a_t = \Theta[\tilde{a}_t] $, where
\begin{equation}\label{action_select}
\tilde{a}_t = 
\begin{cases}
  \cos(\tilde{\mu}_{\theta}(s_t)), &\text{with probability } p_1\\
  -\nabla V(s_t), &\text{with probability }p_2\\
  \cos(\tilde{\mu}_{\theta}(s_t)+\xi_t), &\text{with probability } 1-p_1-p_2\\
\end{cases}
\end{equation}
where $\xi_t\in\Rd$ is sampled from a Gaussian distribution and $p_1\geqslant 0$, $p_2\geqslant 0$, $1-p_1-p_2\geqslant 0$.

We use a replay buffer $\mathcal{R}$ with fixed size $N_{R}$ to store the sampled transitions, where the oldest ones will be discarded when new samples are appended to the buffer. The critic function $Q(s,a)$ can be learned off-policy, allowing us to maintain a large-size replay buffer and sample a mini-batch from the buffer for training neural networks at each step.

\begin{algorithm}[t]
\caption{Reinforcement learning for computing the optimal transition pathway at temperature $\epsilon$.}
\label{alg2}
\vspace{1mm}
Initialize critic and actor networks $Q_{\theta}(s,a)$ and $\mu_{\theta}(s)$.\\
Initialize target networks: $Q'_{\theta}\leftarrow Q_{\theta}$, $\mu'_{\theta}\leftarrow \mu_{\theta}$, and initialize replay buffer $\mathcal{R}$.\\
\For{$step$ = $1$ to $maxstep$}{
    Sample initial state $s_0$ from distribution $p(s)$.\\
    \For{$t$ = $0$ to $maxtime$}{
        Select action $a_t$ according to policy~\eqref{action_select};\\
        Update new state $s_{t+1}=s_t+\gamma\cdot a_t$;\\
        Sample $M$ trajectories starting from $s_{t+1/2}=\frac{1}{2}(s_t+s_{t+1})$ with time $h$ and estimate $F_{\epsilon}(s_{t+1/2})$ as in Eq.~\eqref{MC};\\
        Compute cost $r_t=R_{\epsilon}(s_t,s_{t+1})$;\\
        Store $(s_t,a_t,r_t,s_{t+1})$ in $\mathcal{R}$;\\
        Sample a batch $\mathcal{B}=\{(s_t^i,a_t^i,r_t^i,s_{t+1}^i)\}$  from $\mathcal{R}$;\\
        Compute target values $$y_t^i = r_t^i + Q_{\theta}'(s_{t+1}^i,\mu_{\theta}'(s_{t+1}^i));$$
        Update critic $Q_{\theta}$ using the loss function~\eqref{loss_Q};\\
        Update actor $\mu_{\theta}$ using the sampled policy gradients in Eq.~\eqref{policy_gradient};\\
        Exit if terminal condition $s_{t+1}\in\Omega_{\gamma}$ is met.
    }
    \If{\!\!\!$\mod(step,step_{0})=0$}
    {
    Update target networks $Q'_{\theta}\leftarrow Q_{\theta}$, $\mu'_{\theta}\leftarrow \mu_{\theta}$.
    }
}
{\bf Output}: The critic and actor functions $Q_{\theta}(s,a)$, $\mu_{\theta}(s)$.
\end{algorithm}

{\bf Policy Evaluation.} Target neural networks are often employed in reinforcement learning algorithms to address the instability issue when nonlinear and large scale neural networks are used. Here we duplicate the critic and actor networks as the target networks $Q_{\theta}'(s,a)$, $\mu_{\theta}'(s)$ each time after a particular number of steps. The target networks will be used to compute the target values in the temporal-difference (TD) loss function for training the critic network,
\begin{equation}\label{loss_Q}
\begin{aligned}
    L_{Q}(\theta) &=\frac{1}{|\mathcal{B}|}\sum_{(s_t,a_t,r_t,s_{t+1})\in \mathcal{B}} \left| 
    Q_{\theta}(s_{t},a_t) - y_t
    \right|^2,\\
    y_t&=r_t+ Q_{\theta}'(s_{t+1},\mu_{\theta}'(s_{t+1})).
\end{aligned}
\end{equation}
Here we use the stochastic gradient descent to train the neural networks, and $\mathcal{B}=\{(s_t,a_t,r_t,s_{t+1})\}$ is a batch of transitions sampled from the replay buffer.

{\bf Policy Gradient.} The actor network is trained using the gradient of an average return $J_{\mu}$ over the batch states with respect to the actor parameters by applying the chain rule,
\begin{equation}\label{policy_gradient}
\begin{aligned}
    \nabla_{\theta} J_{\mu} &= \frac{1}{|\mathcal{B}|}\sum_{s_t\in \mathcal{B}} \nabla_a Q_{\theta}(s_t,a)|_{a=\mu_{\theta}(s_t)}\cdot \nabla \mu_{\theta}(s_t).
\end{aligned}
\end{equation}

{\bf Physics-based Learning.} 
General molecular systems are usually invariant to certain transformations of the system, such as  translation, rotation of the configuration and re-indexing of particles of the same species in the system. The physical properties can be described by a transformation function $x'=\mathcal{T}(x)$ mapping a given configuration $x$ and its equivalent ones to a representative configuration $x'$. 
We make the learning process respect the physical properties. In the critic and actor networks, the state-action input $(s,a)$ is transformed into $(s',a')$ by $\mathcal{T}$ before fed into the network. For the actor network, we restore the output $\tilde{a}\in \mathcal{A}$ to $\tilde{a}'\in \mathcal{A}$ by the inverse of $\mathcal{T}$, where $\mathcal{A}$ denotes the action space. Learning over an effective manifold of the configuration space can enhance the efficiency of the algorithm.

A pseudocode for describing the reinforcement learning algorithm is presented in Algorithm~\ref{alg2}. Once the actor network is trained, one can compute a transition pathway with states $\{z_t\}_{0\leqslant t\leqslant T}$ by performing the iterative procedure:
\begin{equation*}
    z_{t+1} = z_t + \gamma\cdot \mu_{\theta}(z_t),\ t\geq 0;\quad z_0=A
\end{equation*}
until the terminal condition $z_{T-1}\in\Omega_{\gamma}$ is satisfied and we set $z_T=B$.

{\bf Remark 1}: {\it The proposed method is able to compute the globally optimal transition pathway. Traditional methods for computing the transition pathway including the nudged elastic band method~\cite{Newell81}, the string method~\cite{weinan2002string,ren2007simplified} and the minimum action method~\cite{weinan2004minimum,heymann2008geometric} are suffering from the issue of metastability, 
as they are performing an iterative procedure in the path space starting from a particular path. The solution relies on the initial guess of path. In the general case with multiple transition pathways (for example, in conformational changes of bio-molecules), suitable initial guess is usually not straightforwardly available. In this case, these methods may get trapped in the neighborhood of a locally optimal solution and produce a path that is far away from the globally optimal one. On the other hand, as in the stochastic policy~\eqref{action_select}, the proposed reinforcement learning algorithm is able to explore the entire configuration space due to the randomness term and focus on the transition region of interest based on the optimal policy introduced by the actor function. The exploration-exploitation balance makes the proposed method able to compute the globally optimal path in the whole configuration space.}

{\bf Remark 2}: {\it There are a number of reinforcement learning algorithms developed recently with substantial improvements over the DDPG algorithm, such as the twin delayed deep deterministic policy gradient algorithm (TD3)~\cite{fujimoto2018addressing} and the soft actor-critic algorithm~\cite{haarnoja2018soft}. 
In this work, we propose a framework of formulating the path-finding problem for dynamical systems and employing RL algorithms to solve the problem. Besides the basic DDPG algorithm, the proposed method enjoys the flexibility of combining with advanced techniques from other RL algorithms, for instance, twin Q-functions and target policy smoothing in TD3, to improve the performance of the method. This might be helpful in specific implementations and will be tested in future works.}

\section{Numerical Examples}\label{Examples}
To illustrate the effectiveness of the proposed method, we apply Algorithm~\ref{alg2} to three benchmark systems including a two-dimensional system, a ten-dimensional system with an extended Mueller potential and the Lennard-Jones system of seven particles. 
The first system is adapted from Ref.~\cite{ren2007simplified} which exhibits two typical transition pathways connecting the metastable states. The latter two systems have been extensively studied in previous works~\cite{khoo2019solving,li2019computing,dellago1998efficient,evans2023computing}. To validate the accuracy of the computed transition pathway, one can compute a reference solution of the transition pathway using the string method or transition tube by constructing the committor function landscape. 
Thus, the three examples can be used to benchmark the proposed method.

We parameterize the critic and actor functions using fully-connected neural networks with two hidden layers and put the hyperbolic tangent function $(\tanh)$ as the activation unit in the network. In the critic network $\tilde{Q}_{\theta}$, as specified in Eq.~\eqref{Q_form}, we put a sigmoid function on the output layer and set $\lambda=1000$. In the hidden actor $\tilde{\mu}_{\theta}$ as in Eq.~\eqref{mu_form}, there is no activation function acting on the output layer. 

In the three examples, we generate episodes in parallel with initial states sampled from a mixed distribution. Specifically, $30\%$ of the initial states are taken as the metastable state $A$ while the remaining are sampled from the equilibrium distribution of the system at a particular temperature $\epsilon'$. The latter part of initial states are collected by simulating the Langevin equation~\eqref{SDE} of the system. In the exploration policy~\eqref{action_select}, we take $p_1 = 1/3$, $p_2 = 1/3$ and $\xi_t$ is sampled from the Gaussian distribution $\mathcal{N}(0,\pi/4)$. The neural networks are trained using the stochastic gradient descent with the Adam optimizer~\cite{kingma2014adam} and batch size $5000$. With a batch of data points sampled from the replay buffer, we train the critic and actor networks repeatedly for $10$ times using the temporal-difference loss function~\eqref{loss_Q} and the sampled policy gradient in Eq.~\eqref{policy_gradient}, respectively. The target networks are updated every $10$ training steps. The parameters in Algorithm~\ref{alg2} used for the three examples are shown in the table in  Appendix~\ref{Parameters}.

\subsection{A Two-dimensional System} 
To illustrate the ability of the method for predicting the globally optimal transition pathway, we first consider a two-dimensional (2D) system with two typical transition pathways. The potential function of the system, as adapted from Ref.~\cite{ren2007simplified}, is given by
\begin{equation}\label{VV}
    V(x,y) = \left[(1-x^2-y^2)^2+\frac{y^2}{x^2+y^2}\right](1+g(y)),
\end{equation}
where $g(y)=1/(1+\exp(-y))$ denotes the sigmoid function. For the system, there are two metastable states at $A=(-1,0)$ and $B=(1,0)$, corresponding to two local minima of the potential $V$.


The system has two minimum energy paths (MEP) connecting $A$ and $B$. We first compute reference solutions for the two paths using the string method~\cite{ren2007simplified}. The string is represented by $501$ points in the 2D plane. We perform the string method twice, in which the initial strings are taken as straight lines connecting the two points $(-0.5,0.5)$ and $(0.5,0.5)$ and connecting $(-0.5,-0.5)$ and $(0.5,-0.5)$, respectively. Plots of the two computed MEPs are shown in the upper panel of Fig.~\ref{fig00}, which resemble upper/lower semicircles connecting $A$ and $B$ in the 2D plane. For convenience, we refer to the two curves as upper/lower MEPs. Also shown in the lower panel of Fig.~\ref{fig00} is the potential energy~\eqref{VV} of the system along the two paths, from which one can observe that the upper MEP has a energy barrier of $\Delta V\simeq 2$, whereas lower MEP has a barrier of $\Delta V\simeq 1$. Therefore, the lower MEP is the globally optimal transition pathway between $A$ and $B$, especially when the magnitude of noise is much less than $1$, {\it i.e.} $\epsilon\ll 1$. 

\begin{figure}[t!]
\centering
\includegraphics[width=\linewidth]{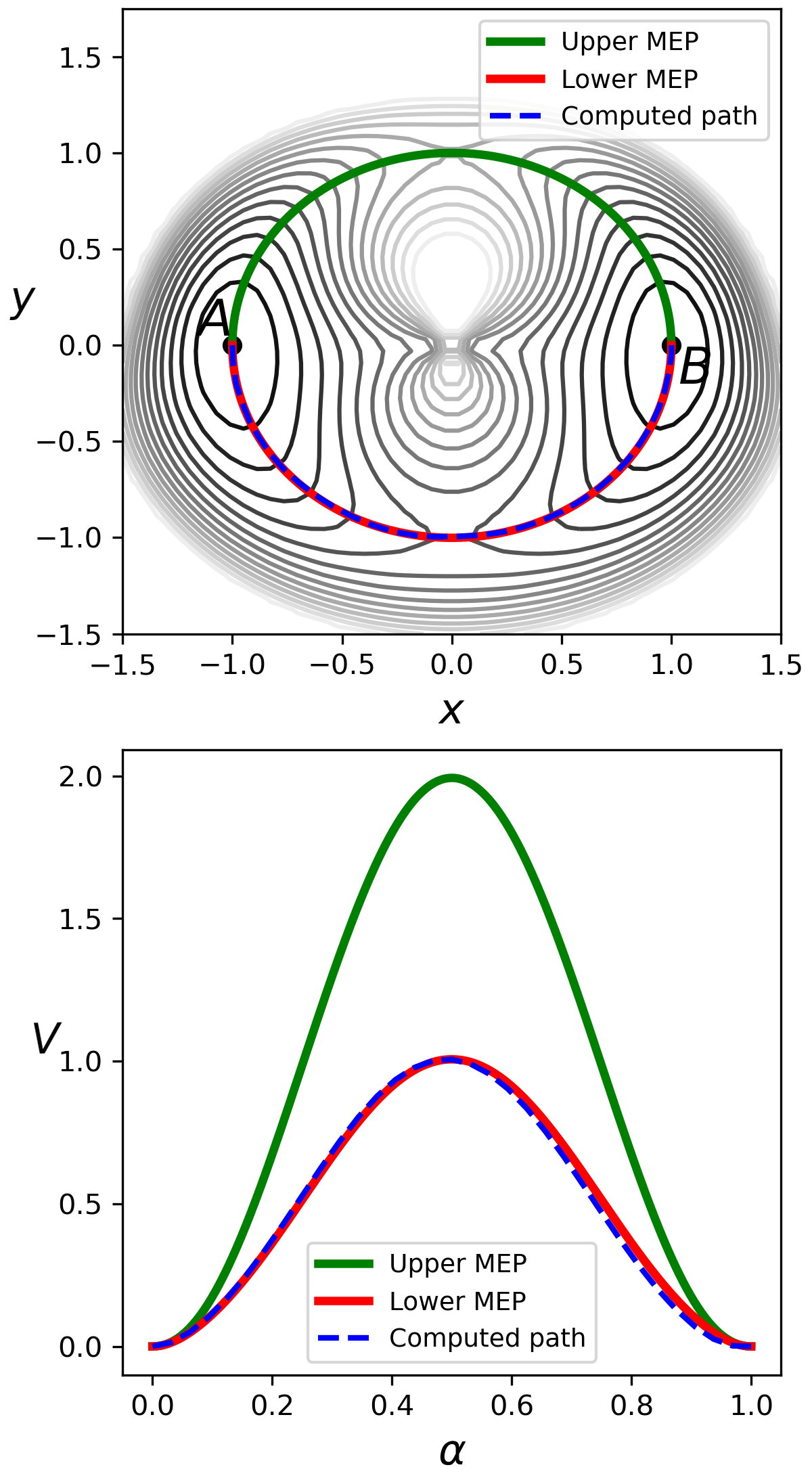}
\caption{Plots of the MEPs computed using the string method and the transition pathway computed from the actor network $\mu_{\theta}(s)$ ({\bf Upper}). Plots of the potential function $V(x)$ along the three paths ({\bf Lower}). The two MEPs are referred to as the upper/lower MEPs. The contour lines in the upper panel indicate the potential function $V(x)$ in Eq.~\eqref{VV}.
}
\label{fig00}
\end{figure}

With the lower MEP, denoted by $\varphi$, one can quantitatively evaluate the path $\varphi_{\theta}$ computed by the proposed method. Note that $\varphi$ is parameterized by the normalized arc-length $\alpha$ of the path. We first re-parameterize $\varphi_{\theta}$ with its normalized arc-length and then evaluate $\varphi_{\theta}$ using the relative error: 
\begin{equation}\label{error}
	e_{\varphi} = \frac{\lVert\varphi_{\theta}(\alpha)-\varphi(\alpha)\rVert}{\lVert\varphi(\alpha)\rVert},
\end{equation}
where the norm is defined as $\lVert\phi(\alpha)\rVert=\sqrt{\frac{1}{n_e}\sum_{i=1}^{n_e} |\phi(i/n_e)|^2}$ for a function $\phi$ using $n_e=100$ discrete points.


We apply Algorithm~\ref{alg2} with $700$ training steps to compute the transition pathway, in $20$ independent runs with random initialization on the critic and actor networks. In the path-finding problem~\eqref{prob2}, we take the path space $C_{\gamma}$ with $\gamma=0.1$ and set $\epsilon=0$. All of the computed paths using the proposed method are close to the lower MEP, as shown in the upper panel of Fig.~\ref{fig00}. The lower panel of the figure shows the potential function along the computed path, which agrees well with the potential along the lower MEP. 
The statistics of relative error for the paths defined in Eq.~\eqref{error} in the $20$ runs is $e_{\varphi}=0.0060\pm0.0020$. The results demonstrate the accuracy of the methods for computing the transition pathway. 

As a comparison, we also apply the string method with random initialization to the system 
in $20$ independent runs. 
On each run, the initial string is taken as a straight line randomly sampled on the 2D plane. Specifically, one end point $x_a$ of the line is sampled from $\mathcal{U}\left([-1.5,0)\times[-1.5,1.5]\right)$ and the other end point $x_b$ is sampled from $\mathcal{U}\left((0,1.5]\times[-1.5,1.5]\right)$, where $\mathcal{U}(\cdot)$ indicates the uniform distribution.
In the total $20$ runs, the string method is convergent to the lower MEP for $11$ runs and to the upper MEP for the remaining runs. 
The results demonstrate that the proposed method statistically outperforms the string method for computing the globally optimal transition pathway.


\subsection{Extended Mueller System}
To illustrate the effectiveness of the method for dealing with high-dimensional systems and rough energy landscapes, we consider the Mueller potential embedded in the ten-dimensional space,
\begin{equation}\label{Example1_V}
	V(x) = V_m(x_1, x_2) + \frac{1}{2\sigma^2}\sum_{i=3}^{10} x_i^2,\quad x \in \mathbb{R}^{10}
\end{equation}
where $V_m(x_1, x_2)$ is the Mueller potential for the first two dimensions,
\begin{equation}
\begin{aligned}
 V_m(x_1, x_2) =& \sum_{i=1}^4 D_i \exp[a_i(x_1 - X_i)^2 \\
 &+ b_i(x_1 - X_i)(x_2 - Y_i) + c_i(x_2 - Y_i)^2] \nonumber \\
&+ \omega \sin(2k\pi x_1)\sin(2k\pi x_2),
\end{aligned}
\end{equation}
and another $8$ harmonic functions are added for the remaining dimensions with the parameter $\sigma$ specifying the strength of the harmonic terms. 
The parameters $\omega$ and $k$ control the roughness of the potential landscape. We take the two metastable states as $A=(-0.558,1.441,0,\dots,0)$ and $B=(0.623,0.028,0,\dots,0)$, which corresponds to two local minimum points of the potential function $V(x)$.

In this example, we take the parameters except $\omega$, $k$ from Ref.~\cite{li2019computing} and consider two cases for the potential function~\eqref{Example1_V} with different values of $\omega$. In the first case ($\omega=0$), the potential landscape is smooth and we apply the method to compute the transition pathway for the system in the zero-temperature limit. In the second case ($\omega>0$), the potential landscape is rather rugged with numerous saddle points and we compute the transition pathway at the temperature $\epsilon=10$.

\begin{figure}[t!]
\centering
\includegraphics[width=\linewidth]{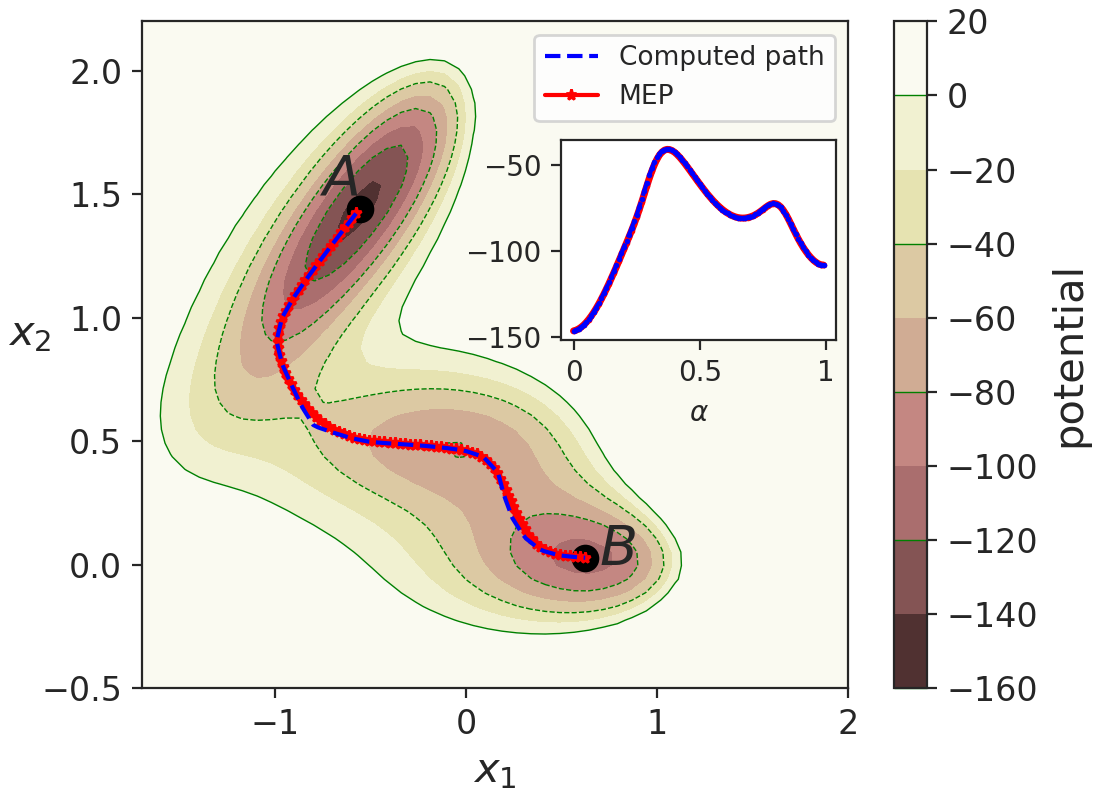}
\caption{Plots of the computed transition pathway $\varphi_{\theta}(\alpha)$ between the metastable states $A$ and $B$ and the minimum energy path (MEP) $\varphi(\alpha)$ which are projected on the $(x_1,x_2)$ plane. The inset plot shows the potential function along the two paths. The contour lines indicate the potential function $V(x)$ in Eq.~\eqref{Example1_V}.}
\label{fig1c}
\end{figure}

\begin{figure}[t!]
	\centering
    \includegraphics[width=.9\linewidth]{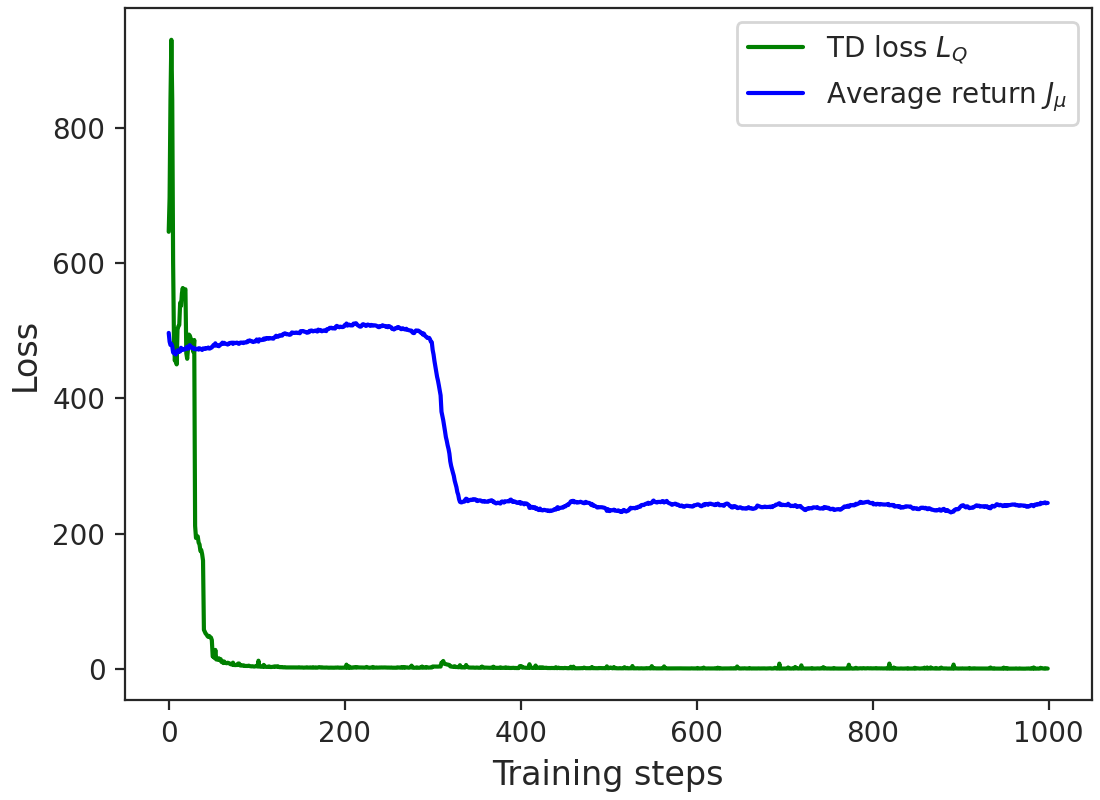}
	\caption{Plots of the temporal-difference (TD) loss function $L_Q$ and the average return $J_{\mu}$ versus the training steps in Algorithm~\ref{alg2}.}
	\label{fig1a}
\end{figure}

\begin{figure}[t!]
	\centering
\includegraphics[width=.9\linewidth]{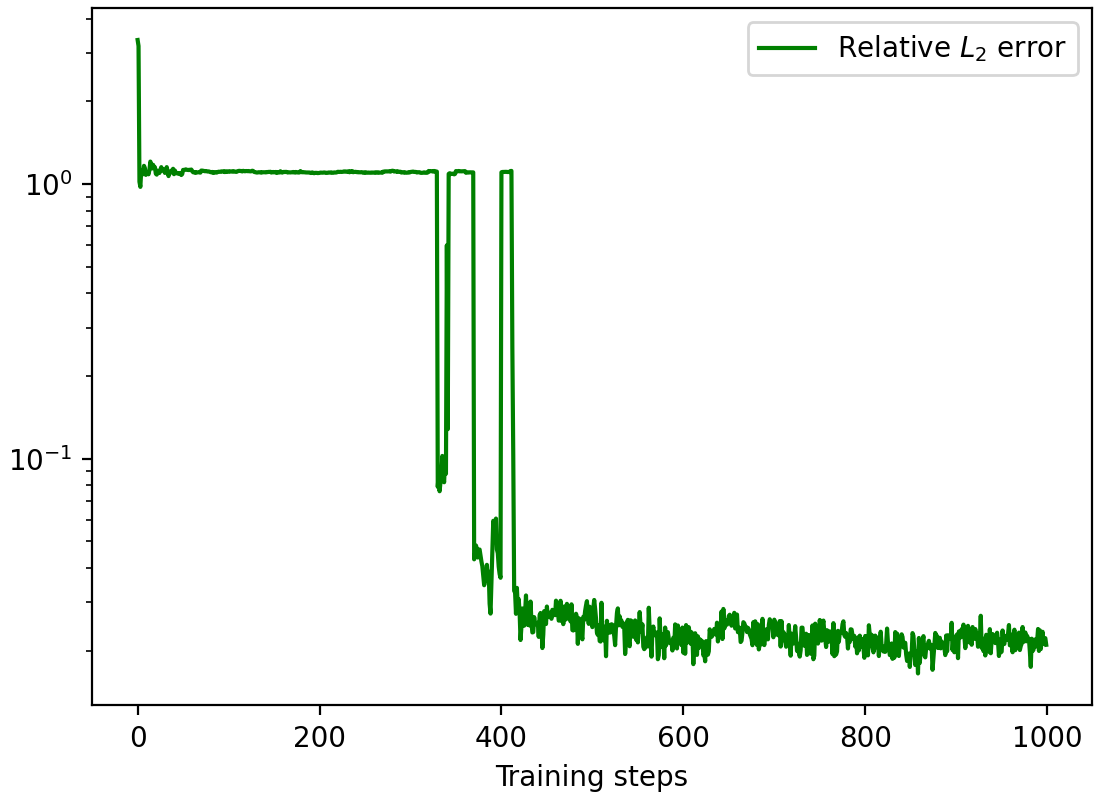}
	\caption{Plot of the relative error for the path $\varphi_{\theta}$ computed from the actor network $\mu_{\theta}(s)$ versus the training steps in Algorithm~\ref{alg2}. The error is defined in Eq.~\eqref{error}.}
	\label{fig1aa}
\end{figure}

\begin{figure}[t!]
	\centering
    \includegraphics[width=\linewidth]{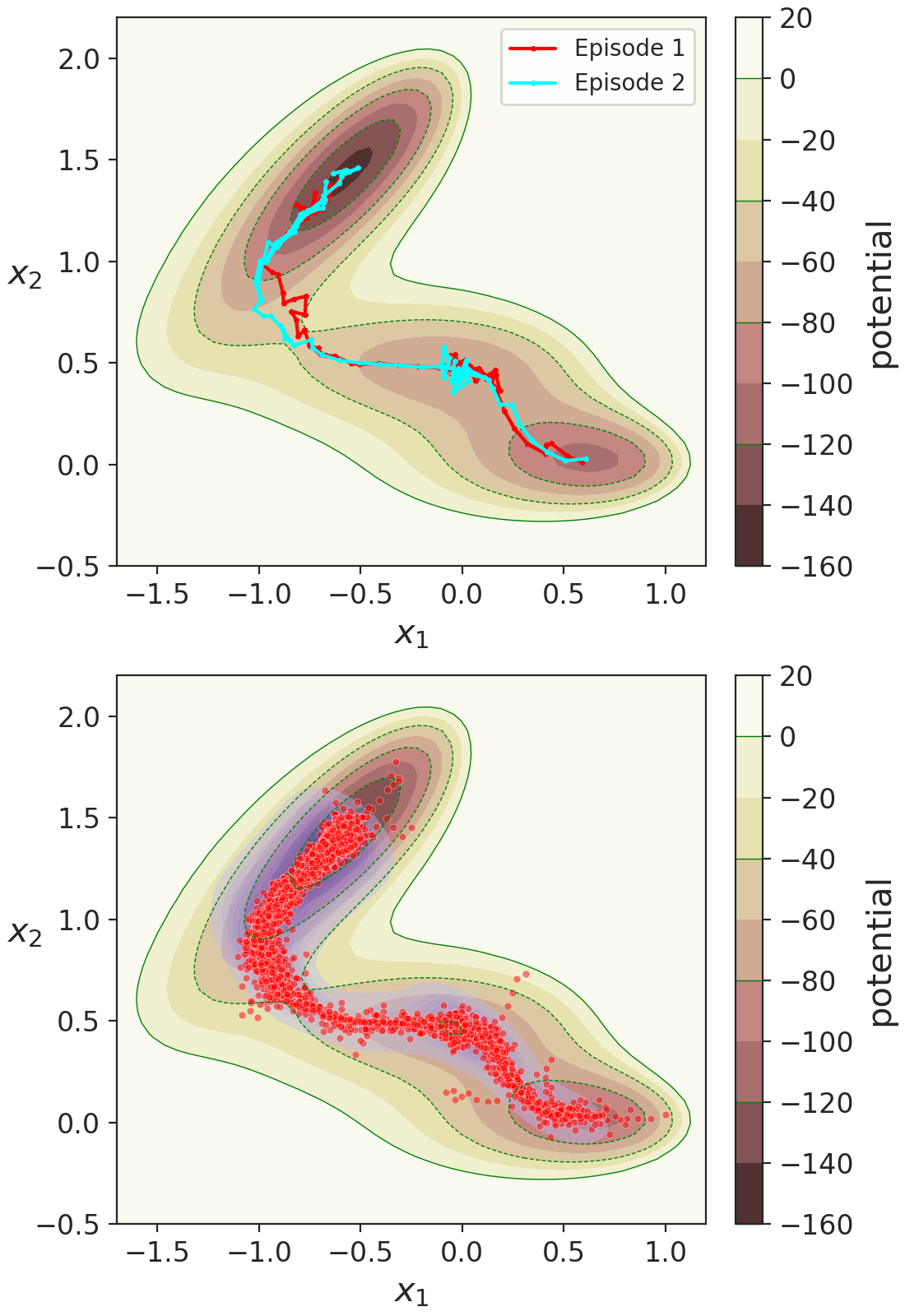}
    \caption{Plot of two generated episodes starting from the state $A$ ({\bf Upper}) and scatter plot of the states $\{s_t^i\}$ in the replay buffer ({\bf Lower}) at the last training step of Algorithm~\ref{alg2}.}
	\label{fig1b}
\end{figure}

{\bf Smooth Mueller Potential.} In the first case, we set the parameter $\omega=0$. For the example, we compute the optimal transition pathway in the path space $\mathbb{C}_{\gamma}$ with $\gamma=0.1$ and apply Algorithm~\ref{alg2} to solve the problem~\eqref{prob2} with $\epsilon=0$. In the algorithm, we conduct $1000$ training steps; at each step, we generate episodes according to the exploration policy~\eqref{action_select} and train the critic and actor networks.


In Fig.~\ref{fig1c}, we plot the transition pathway $\varphi_{\theta}$ computed from the actor network $\mu_{\theta}(s)$, which is projected on the $(x_1,x_2)$ plane. Also shown is the minimum energy path $\varphi$ computed using the string method~\cite{ren2007simplified}. The inset plot in Fig.~\ref{fig1c} shows the potential function along the two paths. In Fig.~\ref{fig1c}, one can observe that the two paths are almost indistinguishable to each other and the method accurately predicts the potential barrier separating the two metastable states via the computed transition pathway. With $\varphi$, the relative error for $\varphi_{\theta}$ as defined in Eq.~\eqref{error} under $10$ independent runs is $e_{\varphi}=0.0203\pm0.0061$. The results demonstrate the accuracy of the proposed reinforcement learning algorithm for predicting the transition pathway of high-dimensional systems. 


In Fig.~\ref{fig1a}, we show plots of the temporal-difference loss function $L_{Q}$ and the average return $J_{\mu}$ over the whole replay buffer versus the training step. 
Also, we show the performance of the actor network, {\it i.e.} the relative error for $\varphi_{\theta}$, versus the training steps in Fig.~\ref{fig1aa}. From the figures, one can observe that the two losses and error are all convergent to low values after about $400$ training steps, which indicates the stability of the algorithm for the system. 
A scatter plot of the states $\{s_t^i\}$ in the replay buffer, projected on the $(x_1,x_2)$ plane, at the last training step of the algorithm is shown in Fig.~\ref{fig1b}, together with two generated episodes starting from the state $A$. 
We observe that the generated data points cluster around the minimum energy path (MEP) as shown in Fig.~\ref{fig1c}, with adequate sampling densities everywhere along the path, and the two episodes are guided 
towards the state $B$ along the MEP with exploration randomness. The results demonstrate that the proposed method is capable of exploring the region regarding transition and sampling transition events efficiently.

\begin{figure}[t!]
\centering
\includegraphics[width=\linewidth]{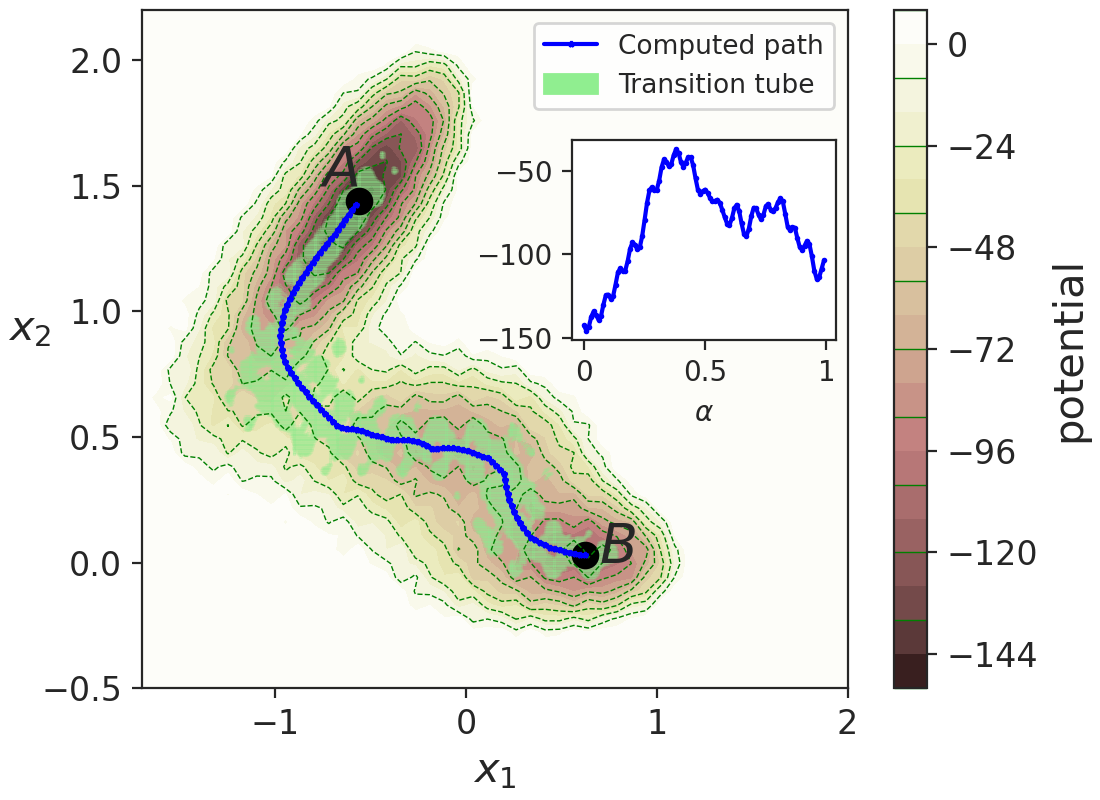}
\caption{Plots of the transition pathway between the metastable states $A$ and $B$ computed from the actor network $\mu_{\theta}(s)$ and the transition tube computed by mapping a committor function landscape of the system, for the rugged potential case at $\epsilon=10$. The inset plot shows the potential function along the computed path.}
\label{fig1d}
\end{figure}

{\bf Rugged Mueller Potential.} In the second case, we set the parameters $\omega=9$, $k=5$ and apply Algorithm~\ref{alg2} to compute the optimal transition pathway at $\epsilon=10$. The parameters in the algorithm are set as those in the previous case. Additionally, we take $h=5\times 10^{-4}$ for generating short trajectories of the system to estimate the function $F_{\epsilon}(z)$ as in Eq.~\eqref{MC}.

Plots of the transition pathway computed from the actor network $\mu_{\theta}(s)$ and the potential function along the path are shown in Fig.~\ref{fig1d}. We validate the solution using a transition tube of the system inside which the transition occurs with high probability~\cite{ren2005transition}. 
Specifically, we compute the transition tube at $\epsilon=10$ by mapping a committor function landscape of the system (See Appendix~\ref{transition_tube}). 
As observed from Fig.~\ref{fig1d}, the computed transition pathway locates near the center of the transition tube, which demonstrates the proposed method is able to predict the transition pathway for high-dimensional systems with rough potential landscapes.

\subsection{Lennard-Jones System}
To illustrate the ability of the algorithm to deal with molecular systems, we apply the method to study a rearrangement process of the Lennard-Jones system, which is a cluster of seven particles on the plane. The system is relatively simple but serves as a good example to benchmark the proposed method.  

In the system, the particles are interacting via the Lennard-Jones potential function
\begin{equation}\label{example2_V}
    V(y_1,\dots,y_7) = \sum_{i<j} 4\epsilon_0 \left[\left(\frac{\sigma_0}{r_{ij}}\right)^{12}-\left(\frac{\sigma_0}{r_{ij}}\right)^{6}\right],
\end{equation}
where $(y_1,\dots,y_7)$ denotes the position-vector of the seven particles, $r_{ij}=|y_i-y_j|$ is the Euclidean distance between particle $i$ and particle $j$ and $\epsilon_0$, $\sigma_0$ specify the energy unit and distance unit in the potential function, respectively. In this example, we take $\epsilon_0=1$ and $\sigma_0=1$. 
In a typical equilibrium state which minimizes the potential~\eqref{example2_V}, the cluster of seven particles forms a hexagon. 
For the system, we are interested in studying the rearrangement process where a particular particle is escaping from the center of the cluster to its surface~\cite{dellago1998efficient,weinan2002string}. Specifically, we look at particle $1$ as the migrating particle during the process. 
Fig.~\ref{fig2a} displays two typical stable configurations of the cluster corresponding to the transition process. We next apply Algorithm~\ref{alg2} to the system for computing a pathway for the transition.

As indicated by the bond-based potential function in Eq.~\eqref{example2_V}, the system is equivalent to any rotation or translation of the cluster, as well as re-indexing of the particles in the system. Based on the properties, we construct a transformation function $\mathcal{T}(x)$ which maps a given configuration $x$ and its equivalent ones to a representative configuration. We incorporate $\mathcal{T}(x)$ into the critic and actor networks to make the learning process reflect the physical properties. The details could be found in Appendix~\ref{tranform}.

In the example, we solve the path-finding problem~\eqref{prob2} over the path space $\mathbb{C}_{\gamma}$ with $\gamma=0.2$ at the temperature $\epsilon=0$. We use the the reinforcement learning algorithm~\ref{alg2} to solve the problem by generating episodes using the policy~\eqref{action_select} and training the critic and actor networks.

The identified transition pathway as computed from the actor network $\mu_{\theta}(s)$ is shown in Fig.~\ref{fig2b}. The results agree well with one transition pathway identified in Ref.~\cite{dellago1998efficient}, which demonstrates that our method is able to predict the transition pathway for the cluster of particles.

\begin{figure}[t!]
\centering
\includegraphics[width=.8\linewidth]{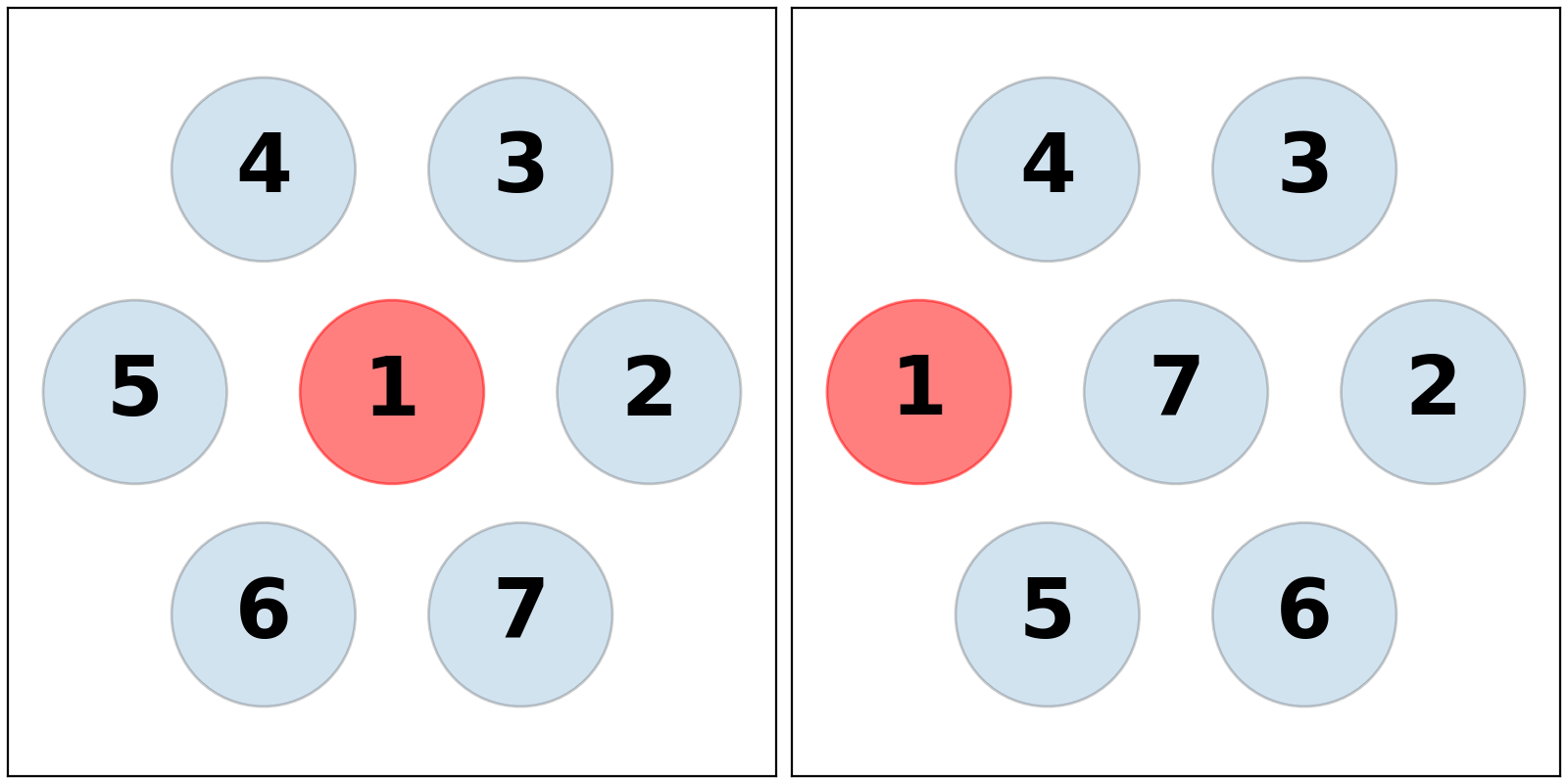}
\caption{Two typical stable configurations of the Lennard-Jones cluster where particle $1$ (in red color) is either at the center ({\bf Left}) or surface ({\bf Right}) of the cluster.}
\label{fig2a}
\end{figure}

\begin{figure}[t!]
\centering
\includegraphics[width=\linewidth]{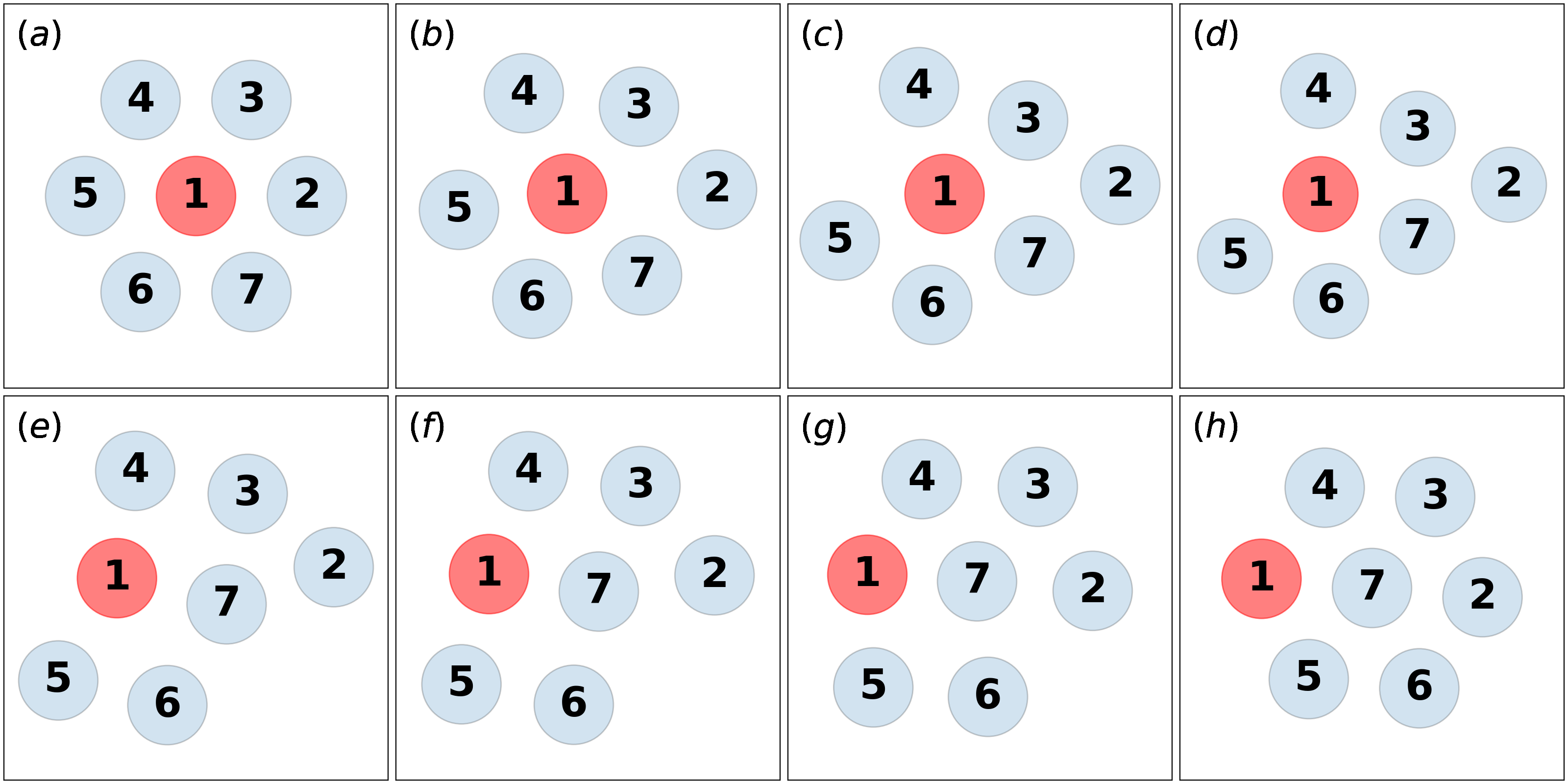}
\caption{Plots of eight states along the computed transition pathway for the Lennard-Jones system (($a$)$\sim$($h$)). During the transition process, particle $1$ is migrating from the center of the cluster to its surface.}
\label{fig2b}
\end{figure}

\section{Conclusions}\label{Conclusion}
In this work, we proposed a deep reinforcement learning algorithm for computing the transition pathway between the metastable states of dynamical systems. It was demonstrated that the proposed method is able to sample the transition events efficiently and thus to predict the globally optimal transition pathway. We illustrated the ability of the method using three model systems.

The proposed method provides a new perspective for investigating the transition mechanism of systems with rough potential energy surfaces. In the future works, we intend to apply the method to more complex systems. 
Another direction for future works is to consider the generalized task of predicting the transition pathway for systems with varying or unseen parameters.


\section*{Acknowledgement}
The work of B. Lin and W. Ren is partially supported by A*STAR under its AME Programmatic programme: Explainable Physics-based AI for Engineering Modelling \& Design (ePAI) [Award No. A20H5b0142]. 

\bibliography{ref}

\begin{thebibliography}{35}
\providecommand{\natexlab}[1]{#1}
\providecommand{\url}[1]{\texttt{#1}}
\expandafter\ifx\csname urlstyle\endcsname\relax
  \providecommand{\doi}[1]{doi: #1}\else
  \providecommand{\doi}{doi: \begingroup \urlstyle{rm}\Url}\fi

\bibitem[Bolhuis et~al.(2002)Bolhuis, Chandler, Dellago, and Geissler]{bolhuis2002transition}
Bolhuis, P.~G., Chandler, D., Dellago, C., and Geissler, P.~L.
\newblock Transition path sampling: Throwing ropes over rough mountain passes, in the dark.
\newblock \emph{Annu. Rev. Phys. Chem.}, 53\penalty0 (1):\penalty0 291--318, 2002.

\bibitem[Dellago et~al.(1998)Dellago, Bolhuis, and Chandler]{dellago1998efficient}
Dellago, C., Bolhuis, P.~G., and Chandler, D.
\newblock Efficient transition path sampling: Application to lennard-jones cluster rearrangements.
\newblock \emph{J. Chem. Phys.}, 108\penalty0 (22):\penalty0 9236--9245, 1998.

\bibitem[Dellago et~al.(2002)Dellago, Bolhuis, and Geissler]{dellago2002transition}
Dellago, C., Bolhuis, P.~G., and Geissler, P.~L.
\newblock Transition path sampling.
\newblock \emph{Adv. Chem. Phys.}, 123:\penalty0 1--78, 2002.

\bibitem[Du et~al.(2021)Du, Li, Li, and Ren]{du2021graph}
Du, Q., Li, T., Li, X., and Ren, W.
\newblock The graph limit of the minimizer of the onsager-machlup functional and its computation.
\newblock \emph{Science China Mathematics}, 64:\penalty0 239--280, 2021.

\bibitem[E \& Vanden-Eijnden(2010)E and Vanden-Eijnden]{vanden2010transition}
E, W. and Vanden-Eijnden, E.
\newblock Transition-path theory and path-finding algorithms for the study of rare events.
\newblock \emph{Annu. Rev. Phys. Chem.}, 61:\penalty0 391--420, 2010.

\bibitem[E et~al.(2002)E, Ren, and Vanden-Eijnden]{weinan2002string}
E, W., Ren, W., and Vanden-Eijnden, E.
\newblock String method for the study of rare events.
\newblock \emph{Phys. Rev. B}, 66\penalty0 (5):\penalty0 052301, 2002.

\bibitem[E et~al.(2004)E, Ren, and Vanden-Eijnden]{weinan2004minimum}
E, W., Ren, W., and Vanden-Eijnden, E.
\newblock Minimum action method for the study of rare events.
\newblock \emph{Commun. Pure Appl. Math.}, 57\penalty0 (5):\penalty0 637--656, 2004.

\bibitem[E et~al.(2005)E, Ren, and Vanden-Eijnden]{weinan2005finite}
E, W., Ren, W., and Vanden-Eijnden, E.
\newblock Finite temperature string method for the study of rare events.
\newblock \emph{J. Phys. Chem. B}, 109\penalty0 (14):\penalty0 6688--6693, 2005.

\bibitem[E et~al.(2007)E, Ren, and Vanden-Eijnden]{ren2007simplified}
E, W., Ren, W., and Vanden-Eijnden, E.
\newblock Simplified and improved string method for computing the minimum energy paths in barrier-crossing events.
\newblock \emph{J. Chem. Phys.}, 126\penalty0 (16), 2007.

\bibitem[Evans et~al.(2023)Evans, Cameron, and Tiwary]{evans2023computing}
Evans, L., Cameron, M.~K., and Tiwary, P.
\newblock Computing committors in collective variables via mahalanobis diffusion maps.
\newblock \emph{Applied and Computational Harmonic Analysis}, 64:\penalty0 62--101, 2023.

\bibitem[Fischer \& Karplus(1992)Fischer and Karplus]{fischer1992conjugate}
Fischer, S. and Karplus, M.
\newblock Conjugate peak refinement: an algorithm for finding reaction paths and accurate transition states in systems with many degrees of freedom.
\newblock \emph{Chem. Phys. Lett.}, 194\penalty0 (3):\penalty0 252--261, 1992.

\bibitem[Fleming(1977)]{fleming1977exit}
Fleming, W.~H.
\newblock Exit probabilities and optimal stochastic control.
\newblock \emph{Applied Mathematics and Optimization}, 4:\penalty0 329--346, 1977.

\bibitem[Freidlin \& Wentzell(2012)Freidlin and Wentzell]{freidlin2012random}
Freidlin, M.~I. and Wentzell, A.~D.
\newblock \emph{Random Perturbations of Dynamical Systems}.
\newblock Springer Press, Berlin, Heidelberg, 2012.

\bibitem[Fujimoto et~al.(2018)Fujimoto, Hoof, and Meger]{fujimoto2018addressing}
Fujimoto, S., Hoof, H., and Meger, D.
\newblock Addressing function approximation error in actor-critic methods.
\newblock In \emph{Proceedings of the 35th International Conference on Machine Learning}, volume~80, pp.\  1587--1596. PMLR, 2018.

\bibitem[Fujisaki et~al.(2010)Fujisaki, Shiga, and Kidera]{fujisaki2010onsager}
Fujisaki, H., Shiga, M., and Kidera, A.
\newblock Onsager--machlup action-based path sampling and its combination with replica exchange for diffusive and multiple pathways.
\newblock \emph{J. Chem. Phys.}, 132\penalty0 (13), 2010.

\bibitem[Grafke \& Vanden-Eijnden(2019)Grafke and Vanden-Eijnden]{grafke2019numerical}
Grafke, T. and Vanden-Eijnden, E.
\newblock Numerical computation of rare events via large deviation theory.
\newblock \emph{Chaos: An Interdisciplinary Journal of Nonlinear Science}, 29\penalty0 (6), 2019.

\bibitem[Guo et~al.(2024)Guo, Gao, Zhang, Han, and Duan]{guo2024deep}
Guo, J., Gao, T., Zhang, P., Han, J., and Duan, J.
\newblock Deep reinforcement learning in finite-horizon to explore the most probable transition pathway.
\newblock \emph{Physica D: Nonlinear Phenomena}, 458:\penalty0 133955, 2024.

\bibitem[Haarnoja et~al.(2018)Haarnoja, Zhou, Abbeel, and Levine]{haarnoja2018soft}
Haarnoja, T., Zhou, A., Abbeel, P., and Levine, S.
\newblock Soft actor-critic: Off-policy maximum entropy deep reinforcement learning with a stochastic actor.
\newblock In \emph{Proceedings of the 35th International Conference on Machine Learning}, volume~80, pp.\  1861--1870. PMLR, 2018.

\bibitem[Heymann \& Vanden-Eijnden(2008)Heymann and Vanden-Eijnden]{heymann2008geometric}
Heymann, M. and Vanden-Eijnden, E.
\newblock The geometric minimum action method: A least action principle on the space of curves.
\newblock \emph{Commun. Pure Appl. Math.}, 61\penalty0 (8):\penalty0 1052--1117, 2008.

\bibitem[J{\'o}nsson et~al.(1998)J{\'o}nsson, Mills, and Jacobsen]{Newell81}
J{\'o}nsson, H., Mills, G., and Jacobsen, K.~W.
\newblock Nudged elastic band method for finding minimum energy paths of transitions.
\newblock In \emph{Classical and quantum dynamics in condensed phase simulations}, pp.\  385--404. World Scientific, 1998.

\bibitem[Khoo et~al.(2019)Khoo, Lu, and Ying]{khoo2019solving}
Khoo, Y., Lu, J., and Ying, L.
\newblock Solving for high-dimensional committor functions using artificial neural networks.
\newblock \emph{Research in the Mathematical Sciences}, 6:\penalty0 1--13, 2019.

\bibitem[Kingma \& Ba(2015)Kingma and Ba]{kingma2014adam}
Kingma, D.~P. and Ba, J.
\newblock Adam: A method for stochastic optimization.
\newblock In \emph{Proceedings of the 3rd International Conference on Learning Representations}, 2015.

\bibitem[Li et~al.(2019)Li, Lin, and Ren]{li2019computing}
Li, Q., Lin, B., and Ren, W.
\newblock Computing committor functions for the study of rare events using deep learning.
\newblock \emph{J. Chem. Phys.}, 151\penalty0 (5), 2019.

\bibitem[Lillicrap et~al.(2015)Lillicrap, Hunt, Pritzel, Heess, Erez, Tassa, Silver, and Wierstra]{lillicrap2015continuous}
Lillicrap, T.~P., Hunt, J.~J., Pritzel, A., Heess, N., Erez, T., Tassa, Y., Silver, D., and Wierstra, D.
\newblock Continuous control with deep reinforcement learning.
\newblock \emph{arXiv preprint arXiv:1509.02971}, 2015.

\bibitem[Maragliano et~al.(2006)Maragliano, Fischer, Vanden-Eijnden, and Ciccotti]{maragliano2006string}
Maragliano, L., Fischer, A., Vanden-Eijnden, E., and Ciccotti, G.
\newblock String method in collective variables: Minimum free energy paths and isocommittor surfacesg.
\newblock \emph{J. Chem. Phys.}, 125\penalty0 (2), 2006.

\bibitem[Mnih et~al.(2013)Mnih, Kavukcuoglu, Silver, Graves, Antonoglou, Wierstra, and Riedmiller]{mnih2013playing}
Mnih, V., Kavukcuoglu, K., Silver, D., Graves, A., Antonoglou, I., Wierstra, D., and Riedmiller, M.
\newblock Playing atari with deep reinforcement learning.
\newblock \emph{arXiv preprint arXiv:1312.5602}, 2013.

\bibitem[Mnih et~al.(2015)Mnih, Kavukcuoglu, Silver, Rusu, Veness, Bellemare, Graves, Riedmiller, Fidjeland, Ostrovski, Petersen, Beattie, Sadik, Antonoglou, King, Kumaran, Wierstra, Legg, and Hassabis]{mnih2015human}
Mnih, V., Kavukcuoglu, K., Silver, D., Rusu, A.~A., Veness, J., Bellemare, M.~G., Graves, A., Riedmiller, M., Fidjeland, A.~K., Ostrovski, G., Petersen, S., Beattie, C., Sadik, A., Antonoglou, I., King, H., Kumaran, D., Wierstra, D., Legg, S., and Hassabis, D.
\newblock Human-level control through deep reinforcement learning.
\newblock \emph{Nature}, 518\penalty0 (7540):\penalty0 529--533, 2015.

\bibitem[Olender \& Elber(1996)Olender and Elber]{olender1996calculation}
Olender, R. and Elber, R.
\newblock Calculation of classical trajectories with a very large time step: Formalism and numerical examples.
\newblock \emph{J. Chem. Phys.}, 105\penalty0 (20):\penalty0 9299--9315, 1996.

\bibitem[Onsager \& Machlup(1953)Onsager and Machlup]{onsager1953fluctuations}
Onsager, L. and Machlup, S.
\newblock Fluctuations and irreversible processes.
\newblock \emph{Physical Review}, 91\penalty0 (6):\penalty0 1505, 1953.

\bibitem[Ren et~al.(2005)Ren, Vanden-Eijnden, Maragakis, and E]{ren2005transition}
Ren, W., Vanden-Eijnden, E., Maragakis, P., and E, W.
\newblock Transition pathways in complex systems: Application of the finite-temperature string method to the alanine dipeptide.
\newblock \emph{J. Chem. Phys.}, 123\penalty0 (13):\penalty0 6688--6693, 2005.

\bibitem[Rotskoff et~al.(2022)Rotskoff, Mitchell, and Vanden-Eijnden]{rotskoff2022active}
Rotskoff, G.~M., Mitchell, A.~R., and Vanden-Eijnden, E.
\newblock Active importance sampling for variational objectives dominated by rare events: Consequences for optimization and generalization.
\newblock In \emph{Mathematical and Scientific Machine Learning}, pp.\  757--780, 2022.

\bibitem[Silver et~al.(2014)Silver, Lever, Heess, Degris, Wierstra, and Riedmiller]{silver2014deterministic}
Silver, D., Lever, G., Heess, N., Degris, T., Wierstra, D., and Riedmiller, M.
\newblock Deterministic policy gradient algorithms.
\newblock In \emph{Proceedings of the 31st International Conference on Machine Learning}, volume~32, pp.\  387--395, 2014.

\bibitem[Voter(1997)]{voter1997hyperdynamics}
Voter, A.~F.
\newblock Hyperdynamics: Accelerated molecular dynamics of infrequent events.
\newblock \emph{Phys. Rev. Lett.}, 78\penalty0 (20):\penalty0 3908, 1997.

\bibitem[Wang et~al.(2010)Wang, Zhang, and Wang]{wang2010kinetic}
Wang, J., Zhang, K., and Wang, E.
\newblock Kinetic paths, time scale, and underlying landscapes: A path integral framework to study global natures of nonequilibrium systems and networks.
\newblock \emph{J. Chem. Phys.}, 133\penalty0 (12), 2010.

\bibitem[Zhou et~al.(2008)Zhou, Ren, and E]{zhou2008adaptive}
Zhou, X., Ren, W., and E, W.
\newblock Adaptive minimum action method for the study of rare events.
\newblock \emph{J. Chem. Phys.}, 128\penalty0 (10), 2008.

\end{thebibliography}
\bibliographystyle{icml2024}

\newpage
\appendix
\onecolumn

\section{The path space $\mathbb{C}_{\gamma}$.}\label{function_space}

\begin{theorem}
    The set $\bigcup_{\gamma>0}\mathbb{C}_{\gamma}$ is a dense subset of $\bigcup_{T>0}\mathbb{C}_{[0,T]}$. 
\end{theorem}
\begin{proof}
For a path $\varphi(t)$ in $\bigcup_{\gamma>0}\mathbb{C}_{\gamma}$ which is continuous and represented by a finite number of line segments, one can easily see that the path is absolute continuous. Thus,  $\bigcup_{\gamma>0}\mathbb{C}_{\gamma}$ is a subset of $ \bigcup_{T>0}\mathbb{C}_{[0,T]}$.

Suppose $\varphi(t) \in \mathbb{C}_{[0, T]}$ for a number $T>0$. Given $\epsilon > 0$, we shall prove that there exists $\gamma > 0$ and a $\psi(t) \in \mathbb{C}_{\gamma}$ such that
\begin{equation}
	\max_{t\in [0,T]}| \psi(t) - \varphi(t) | < 3\epsilon.
\end{equation}
In the following, when we consider a polygonal curve, we assume the time derivative of the curve on each line segment is constant.

Since $\varphi$ is absolute continuous, it is uniformly continuous. Thus, there exists $\delta > 0$ such that 
\begin{equation}\label{uniform_cont}
	|\varphi(t_1) - \varphi(t_2)| < \epsilon,\quad \forall\ |t_1 - t_2| < \delta.
\end{equation}
Let $0 = t_0 < t_1 < \cdots < t_n = T $ such that $|t_{i} - t_{i - 1}| < \delta$ and $\varphi(t_{i-1})\neq\varphi(t_i)$ for $1 \leqslant i \leqslant n$. Then $|\varphi(t_i) - \varphi(t_{i - 1})| < \epsilon$ for $1 \leqslant i \leqslant n$. 

We first construct a polygonal curve $\psi_0(t),\ 0\leqslant t\leqslant T$ such that
\begin{equation}\label{Eq11}
	\max_{t\in [0,T]}|\psi_0(t) - \varphi(t)| < 2\epsilon.
\end{equation}
Let $\psi_0(t)$ be a polygonal curve with vertices at the times $\{t_i\}_{0\leqslant i \leqslant n}$ and
\begin{equation}\label{def_psi_0}
	\psi_0(t_{i}) = \varphi(t_i),\quad 0\leqslant i \leqslant n.
\end{equation}

Then for arbitrary time interval $[t_{i - 1}, t_{i}]$ and $t\in[t_{i - 1}, t_{i}]$, we have
\begin{equation*}
\begin{aligned}
    |\psi_0(t) - \varphi(t)| & \leqslant |\psi_0(t) - \psi_0(t_{i - 1})| + |\psi_0(t_{i - 1}) - \varphi(t_{i - 1})| + |\varphi(t_{i - 1}) - \varphi(t)|\\
	&=  |\psi_0(t) - \psi_0(t_{i - 1})| + |\varphi(t_{i - 1}) - \varphi(t)| \\
	& \leqslant |\psi_0(t_i) - \psi_0(t_{i - 1})| + |\varphi(t_{i - 1}) - \varphi(t)| \\
	& = |\varphi(t_i) - \varphi(t_{i - 1})| + |\varphi(t_{i - 1}) - \varphi(t)|\\
	& < 2\epsilon.
\end{aligned}
\end{equation*}
The last inequality is due to the uniform continuity~\eqref{uniform_cont} of $\varphi$. Hence, we have proved the inequality~\eqref{Eq11}.

Next, we construct a polygonal curve $\psi(t),\ 0\leqslant t\leqslant T$ with line segments of uniform length such that 
\begin{equation}\label{Eq12}
	\max_{t\in [0,T]}| \psi(t)-\psi_0(t) | < \epsilon.
\end{equation}
Denote the length of line segments in $\psi_0$ by
\begin{equation*}
	\gamma_i = |\psi_0(t_{i}) - \psi_0(t_{i - 1})|,\quad 1\leqslant i \leqslant n
\end{equation*}
and the minimum length by
\begin{equation}\label{min_len}
	\gamma = \min_{1\leqslant i \leqslant n}|\psi_0(t_{i}) - \psi_0(t_{i - 1})|.
\end{equation}
Since $\gamma_i = |\varphi(t_{i}) - \varphi(t_{i - 1})|>0$, we have $\gamma>0$. By the definition~\eqref{def_psi_0} of $\psi_0$ and the uniform continuity~\eqref{uniform_cont} of $\varphi$, we have $\gamma < \epsilon$. We write the segment length in $\psi_0$ as a sum of a multiple of the minimum length and a residual value, {\it i.e.}
\begin{equation*}
	\gamma_i = k_i \gamma + 2l_i,\quad \text{where}\ k_i \in \mathbb{N}^*\ \text{and}\ 0 \leqslant 2 l_i < \gamma.
\end{equation*}

To construct the polygonal curve $\psi$, we first let $\psi(t_i) = \psi_0(t_i)$ for $0\leqslant i \leqslant n$. Then we shall carefully define $\psi$ on each interval $[t_{i - 1}, t_i]$ for $1\leqslant i \leqslant n$. 

We first consider the case where $l_i>0$. Denote the time slice $\Delta t = t_i - t_{i - 1}$. Let $\beta$ be a unit vector that is normal to the $i$th line segment of $\psi_0$, {\it i.e.}
\begin{equation*}
	\langle \beta, \psi_0(t_{i - 1}) - \psi_0(t_{i})\rangle = 0.
\end{equation*}
On the interval $[t_{i - 1}, t_i]$, we let the polygon curve $\psi$ to have vertices $\{ \psi(t_{i - 1}), \psi(t_{i -1} + \dfrac{l_i}{\gamma_i}  \Delta t), \psi(t_{i} - \dfrac{l_i}{\gamma_i}  \Delta t), \psi(t_i)\}$ where
\begin{equation*}
\begin{aligned}
    \psi(t_{i -1} + \frac{l_i}{\gamma_i}  \Delta t) &= \psi_0(t_{i -1} + \frac{l_i}{\gamma_i}  \Delta t) + \sqrt{\gamma^2 - l_i^2}\beta, \\ 
	\psi(t_{i } - \frac{l_i}{\gamma_i}  \Delta t) &= \psi_0(t_{i} - \frac{l_i}{\gamma_i}  \Delta t) + \sqrt{\gamma^2 - l_i^2}\beta.
\end{aligned}
\end{equation*}
The specific form of $\psi$ on $[t_{i-1},t_i]$ is defined as
\begin{equation}\label{def_psi}
\psi(t_{i - 1} + t) =
\begin{cases}
	 \psi_0(t_{i - 1} + t) + \dfrac{\gamma_i t}{l_i\Delta t} \sqrt{\gamma^2 - l_i^2}\beta, &t \in [0, \dfrac{l_i}{\gamma_i}  \Delta t]; \\
	 \psi_0(t_{i - 1} + t) + \sqrt{\gamma^2 - l_i^2}\beta, &t \in [\dfrac{l_i}{\gamma_i}  \Delta t, \Delta t - \dfrac{l_i}{\gamma_i}  \Delta t];\\
	 \psi_0(t_{i - 1} + t) + \dfrac{\gamma_i (\Delta t - t)}{l_i\Delta t} \sqrt{\gamma^2 - l_i^2}\beta, &t \in [\Delta t - \dfrac{l_i}{\gamma_i}  \Delta t, \Delta t].
\end{cases}
\end{equation}
One can verify that the following distance is equal to the minimum length $\gamma$ as defined in Eq.~\eqref{min_len},
\begin{equation*}
\begin{aligned}
   &\quad\left|\psi(t_{i - 1} + \frac{l_i}{\gamma_i} \Delta t ) - \psi(t_{i - 1})\right|\\
	&=\left|\psi_0(t_{i - 1} + \frac{l_i}{\gamma_i} \Delta t ) - \psi_0(t_{i - 1}) + \sqrt{\gamma^2 - l_i^2}\beta\right|\\
	&=\left( |\psi_0(t_{i - 1} + \frac{l_i}{\gamma_i} \Delta t ) - \psi_0(t_{i - 1})|^2 + (\gamma^2 - l_i^2)  \right)^{\frac{1}{2}} \\
	&= \left( \left(\frac{l_i}{\gamma_i}\Delta t  \frac{|\psi_0(t_i) - \psi_0(t_{i  - 1})|}{\Delta t} \right)^2 + (\gamma^2 - l_i^2)  \right)^{\frac{1}{2}}\\
	&= \gamma. 
\end{aligned}
\end{equation*}
Similarly, we have 
\begin{equation*}
	\left|\psi(t_{i} - \frac{l_i}{\gamma_i} \Delta t ) - \psi(t_{i})\right| = \gamma.
\end{equation*}
Also, one can show that the distance between $\psi(t_{i -1} + \dfrac{l_i}{\gamma_i}  \Delta t)$ and $\psi(t_{i} - \dfrac{l_i}{\gamma_i}  \Delta t)$ is a multiple of $\gamma$,
\begin{equation*}
\begin{aligned}
	&\quad\left| \psi(t_{i} - \frac{l_i}{\gamma_i} \Delta t ) - \psi(t_{i - 1} + \frac{l_i}{\gamma_i} \Delta t) \right| \\
	&=\left| \psi_0(t_{i} - \frac{l_i}{\gamma_i} \Delta t ) - \psi_0(t_{i - 1} + \frac{l_i}{\gamma_i} \Delta t) \right|\\
	&= \frac{t_i - t_{i - 1} - \frac{2l_i}{\gamma_i}\Delta t}{\Delta t} |\psi_0(t_{i}) - \psi_0(t_{i - 1})|\\
	&= \frac{\Delta t - \frac{2l_i}{\gamma_i}\Delta t}{\Delta t} \gamma_i \\
	&= k_i \gamma.
\end{aligned}
\end{equation*}
In sum, the curve $\psi$ over $[t_{t-1},t_i]$ as defined in Eq.~\eqref{def_psi} is composed of three line segments with the lengths $\gamma$, $k_i\gamma$ and $\gamma$, respectively. One can treat $\psi$ over $[t_{t-1},t_i]$ as a polygon curve of $k_i+2$ line segments with uniform length $\gamma$. Thus, we see that the constructed curve $\psi\in \mathbb{C}_{\gamma}$.

For another case where $l_i=0$, we simply set $\psi(t)=\psi_0(t)$, $t\in[t_{i-1},t_i]$. Then $\psi$ over $[t_{i-1},t_i]$ can be regarded as a polygon curve of $k_i$ line segments with uniform length $\gamma$. Also in this case, the constructed curve $\psi\in \mathbb{C}_{\gamma}$.

Moreover, in the latter case $l_i=0$, we have $\max_{t \in [t_{i - 1}, t_i]} |\psi(t) - \psi_0(t)|=0$. In the case $l_i>0$, we have
\begin{equation*}
	\max_{t \in [t_{i - 1}, t_i]} |\psi(t) - \psi_0(t)| = \left|\sqrt{\gamma^2 - l_i^2} \beta\right| \leqslant \gamma < \epsilon.
\end{equation*}

This proves the inequality~\eqref{Eq12}. Therefore, combing the inequality~\eqref{Eq11} and~\eqref{Eq12}, we find a path $\psi\in\mathbb{C}_{\gamma}$ such that
\begin{equation*}
\begin{aligned}
	\max_{t\in [0,T]} | \psi(t) - \varphi(t) |& \leqslant \max_{t\in [0,T]} | \psi(t) - \psi_0(t) | + \max_{t\in [0,T]} | \psi_0(t) - \varphi(t) |\\
	& < \epsilon + 2 \epsilon \\
	& = 3 \epsilon.
\end{aligned}
\end{equation*}
\end{proof}

\section{Numerical quadratures for approximating the integral in problem~\eqref{min_Sr}.}\label{quadrature}

One can use a numerical quadrature with $m$ grid points to approximate the integral as in problem~\eqref{min_Sr}, where $m$ is a positive integer. We discretize the time interval $[0,h_i]$ using $m$ uniform points: $t_j=(j-1/2)/m\cdot h_i$, $1\leqslant j\leqslant m$. Denote the position of the path $\varphi^i(t)$ at time $t_j$ by \begin{equation*}
    z_i^j=\varphi^i(t_j)=z_i+\frac{j-1/2}{m}(z_{i+1}-z_i),\quad 1\leqslant j\leqslant m.
\end{equation*} 
Then the integral in the problem~\eqref{min_Sr} can be approximated by
\begin{align}
     \tilde{L}(z_i,z_{i+1};h_i) &\approx 
    \frac{h_i}{m}\sum_{j=1}^m \left| \frac{z_{i+1}-z_i}{h_i}+\nabla V(z_i^j)\right|^2\nonumber\\
    &\geqslant 2 |z_{i+1}-z_i|\cdot
    \sqrt{\frac{1}{m}\sum_{j=1}^m |\nabla V(z_i^j)|^2}
    +2 \langle z_{i+1}-z_i , \frac{1}{m}\sum_{j=1}^m \nabla V(z_i^j)\rangle.\label{cost_general1}
\end{align}
where the minimum value is achieved by setting
\begin{equation*}
    h_i^* = |z_{i+1}-z_i|\left/
    \sqrt{\frac{1}{m}\sum_{j=1}^m |\nabla V(z_i^j)|^2}\right..
\end{equation*}
By setting $m=1$, the cost function~\eqref{cost_general1} reduces to the one in Eq.~\eqref{cost1} using the mid-point quadrature as in the paper.

\newpage
\section{The parameters in Algorithm~\ref{alg2} used for the numerical examples in Section~\ref{Examples}.}\label{Parameters}

\begin{table}[!h]
	\caption{The parameters in Algorithm~\ref{alg2} used for the 2D system, extended Mueller system and Lennard-Jones system of seven particles.}
    \vspace{5mm}
	\label{tab1a}
	\begin{center}
		\begin{tabular}{ ccccc }
			\Xhline{2\arrayrulewidth} \vspace{-0.25cm}\\
			& Parameters  & 2D system & Mueller system & Lennard-Jones system 
            \vspace{.1cm}\\
			\Xhline{2\arrayrulewidth} \vspace{-0.25cm}\\

            \multirow{4}{*}{{\rotatebox[origin=c]{90}{Critic $\tilde{Q}_{\theta}$}}} 
            & Network structure & $4$-$50$-$50$-$1$ & $20$-$50$-$50$-$1$  & $24$-$100$-$100$-$1$ \\
            & Activation on hidden layers & $\tanh$ & $\tanh$  & $\tanh$ \\
            & Activation on output layer  & sigmoid & sigmoid  & sigmoid \\
            & $\lambda$  & $1000$ & $1000$  & $1000$
			\vspace{.1cm}\\
            \hline \vspace{-0.25cm}\\

            \multirow{3}{*}{{\rotatebox[origin=c]{90}{Actor $\tilde{\mu}_{\theta}$}}} 
            & Network structure & $2$-$50$-$50$-$2$ & $10$-$50$-$50$-$10$  & $12$-$100$-$100$-$12$ \\
            & Activation on hidden layers & $\tanh$ & $\tanh$  & $\tanh$ \\
            & Activation on output layer  & None & None  & None
			\vspace{.1cm}\\
            \hline \vspace{-0.25cm}\\

            \multirow{3}{*}{{\rotatebox[origin=c]{90}{Problem}}} 
            & $\gamma$ in the path space $\mathbb{C}_{\gamma}$ & $0.1$ & $0.1$  & $0.2$ \\
            & Numerical quadrature in cost $R_{\epsilon}$ & mid-point scheme & mid-point scheme & mid-point scheme \\
            & $(M,h)$ in estimating $F_{\epsilon}$  & - & $(1000,5\times 10^{-4})$  & - 
			\vspace{.1cm}\\
            \hline \vspace{-0.25cm}\\
            
			\multirow{8}{*}{{\rotatebox[origin=c]{90}{Algorithm}}} 
            & $maxstep$ (\# of training steps) & $700$ & $1000$  & $1000$ \\
            & $maxtime$ (maximum length of one episode) & $50$ & $100$  & $100$ \\
            & \# of episodes per training step & $50$ & $50$  & $100$ \\
            & Temperature $\epsilon'$ for sampling initial states & $0.3$ & $20$  & $0.3$ 
            \\
            & $(p_1,p_2)$ in exploration policy & $(1/3,1/3)$ & $(1/3,1/3)$ & $(1/3,1/3)$ 
            \\
            & Distribution for $\xi_t$ in exploration policy & $\mathcal{N}(0,\pi/4)$ & $\mathcal{N}(0,\pi/4)$ & $\mathcal{N}(0,\pi/4)$
			\\
            & $step_0$ for target networks  & $10$ & $10$ & $10$ 
			\\
			& buffer size $N_R$ & $10^5$ & $10^5$ & $10^6$ 
			\vspace{.1cm}\\
            \hline \vspace{-0.25cm}\\
            
            \multirow{3}{*}{{\rotatebox[origin=c]{90}{Training}}} 
            & Optimizer  & Adam & Adam  & Adam
			\\
			& Learning rate & $10^{-3}$ & $10^{-3}$ & $10^{-3}$
            \\
            & Batch size $|\mathcal{B}|$ & $5000$ & $5000$ & $5000$
            \vspace{.1cm}\\
			\Xhline{2\arrayrulewidth}
		\end{tabular}
	\end{center}
\end{table}

\section{Computing the transition tube for the Mueller system.}\label{transition_tube}
For the Mueller system with potential function~\eqref{Example1_V}, we define the two metastable set $S_A$, $S_B$ with radius $r_0$ around the states $A$, $B$ as 
\begin{equation*}
    S_A = \{x\in\R^{10}:|(x_1,x_2)-(A_1,A_2)|<r_0\},\quad
    S_B = \{x\in\R^{10}:|(x_1,x_2)-(B_1,B_2)|<r_0\}
\end{equation*}
with $r_0=0.1$, where $(A_1,A_2)$ and $(B_1,B_2)$ denote the first two numbers in the coordinates of $A$ and $B$, respectively. The first hitting time of the two sets $S_A$ and $S_B$ is defined as
\begin{equation*}
    \tau_A(z) = \inf_{t>0} \{x_t\in S_A: x_0=z\},\quad
    \tau_B(z) = \inf_{t>0} \{x_t\in S_B: x_0=z\}.
\end{equation*}
The committor function $q(x)$ is defined in the configuration space, which is the probability that the system starting from $x$ first arrives in $S_B$ rather than $S_A$:
\begin{equation*}
    q(x) = \text{Prob}[\tau_A(x)>\tau_B(x)].
\end{equation*}
A mathematical description for the committor function is given by the backward Kolmogorov equation with the Dirichlet boundary conditions:
\begin{equation}\label{PDE}
    \begin{cases}
    \nabla V(x)\cdot \nabla q(x) - \epsilon\Delta q(x) = 0,\quad x\in\Omega\setminus(S_A\cup S_B)\\
    q(x)=0,\ x\in \partial S_A;\quad q(x)=1,\ x\in\partial S_B.
\end{cases}
\end{equation}
We compute the committor function $q_m(x_1,x_2)$ for the $2$D Mueller system with the potential $V_m(x_1,x_2)$ at $\epsilon=10$ by solving the partial differential equation~\eqref{PDE} using the finite element method. The computational domain is taken as $\Omega=[-1.5,1]\times[-0.5,2]$. Then the committor function $q(x)$ for the $10$D Mueller system with the potential~\eqref{Example1_V} is given by $q(x_1,\dots,x_{10})=q_m(x_1,x_2)$. 

The committor function itself is a good reactive coordinate for describing the transition of the system. As illustrated in Ref.~\cite{ren2005transition}, the transition tube can be characterized by the iso-committor surfaces of the committor function. For the $10$D Mueller system considered in the numerical example, the transition events can be characterized by the transition tube corresponding to the first two dimensions $(x_1,x_2)$.

Next, we compute the transition tube for the two-dimensional potential $V_m(x_1,x_2)$ at $\epsilon=10$. 
To approximate the iso-committor surfaces, we divide the configuration space into sub-regions using a transition pathway ({\it e.g.} the minimum energy path corresponding to the potential $V_m(x_1,x_2)$ with $\omega=0$). Specifically, the transition pathway $\varphi$ is represented by $N_p=185$ states $z_1,\dots,z_{N_p}\in \R^2$ with equal distance. Denote the committor function $q_m(x_1,x_2)$ on these points by $q_i=q_m(z_i)$, $1\leq i\leq N_p$. Note that $q_1<\cdots<q_{N_p}$. We approximate the $q_i$-isocommittor surface using the region 
\begin{equation*}
\Omega_i=\{(x_1,x_2)\in\R^2:(q_{i-1}+q_{i})/2<q_m(x_1,x_2)<(q_i+q_{i+1})/2\}, \quad 1\leq i\leq N_p
\end{equation*}
where $q_0=-q_1$ and $q_{N_p+1}=2-q_{N_p}$.

\begin{figure}[t!]
\centering
\includegraphics[width=\linewidth]{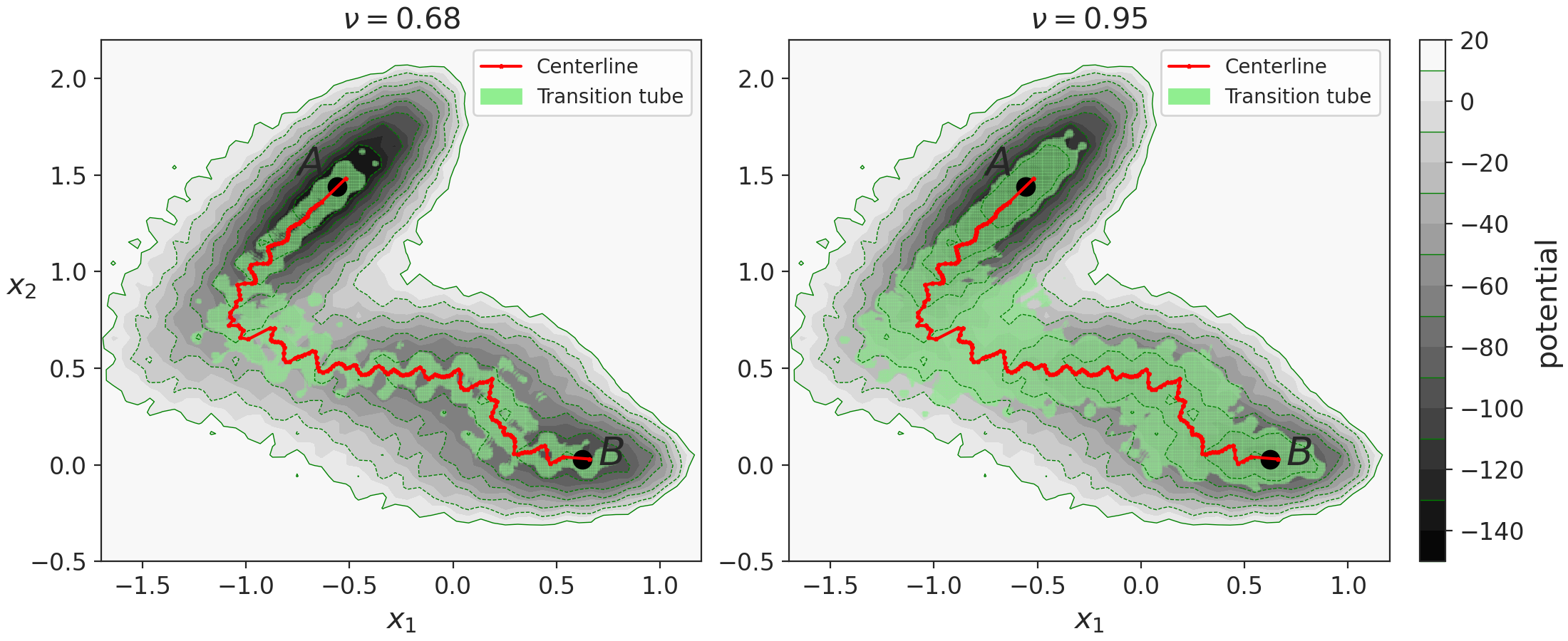}
\caption{Plots of the computed transition tube and its centerline for the potential function $V_m(x_1,x_2)$ at $\epsilon=10$ with different confidence levels, $\nu=0.68$ ({\bf Left}) and $\nu=0.95$ ({\bf Right}). The contour lines indicate the potential function.}
\label{figA1}
\end{figure}

We discretize the computational domain $\Omega$ using a mesh with $500\times 500$ grid points. The set of grid points is denoted by $\{X_k\}_{k\in \mathcal{I}}$. Denote by $\mathcal{I}_i$ the indices of the grid points in $\{X_k\}_{k\in \mathcal{I}}$ which locate inside the region $\Omega_i$,
\begin{equation*}
    \mathcal{I}_i = \{k\in\mathcal{I}:X_k\in\Omega_i\}\quad 1\leq i\leq N_p.
\end{equation*}
We assign the following Gibbs-Boltzmann weight to each grid point in $\{X_k\}_{k\in\mathcal{I}_i}$ in the region $\Omega_i$,
\begin{equation*}
    w_k = \frac{1}{Z_i}\exp\left[-\frac{V_m(X_k)}{\epsilon}\right],\quad k\in\mathcal{I}_i
\end{equation*}
where the normalization constant $Z_i=\sum_{k\in\mathcal{I}_i} \exp[-V_m(X_k)/\epsilon]$. The centerline of the transition tube can be represented by the weighted average of the grid points on each region $\Omega_i$:
\begin{equation*}
    c_i = \sum_{k\in\mathcal{I}_i} w_k \cdot X_k,\quad 1\leq i\leq N_p.
\end{equation*}
To represent the transition tube, we sort the weights and choose a subset $\{X_k\}_{k\in\mathcal{J}_i}$ of $\{X_k\}_{k\in\mathcal{I}_i}$ containing the least number of gird points with the largest weights $w_k$ such that
\begin{equation*}
    \sum_{k\in \mathcal{J}_i} w_k \geq\nu,
\end{equation*}
where $\nu\in[0,1]$ denotes the confidence level. Then the transition tube is represented by a collection of grid point sets, $\{X_k\}_{k\in\mathcal{J}_i}$, $1\leq i\leq N_p$.

In Fig.~\ref{figA1}, we show the plots of the computed transition tube and its centerline using two different values for $\nu$ ($\nu=0.68$ and $\nu=0.95$). In the numerical example, we take $\nu=0.68$, as shown in upper panel of Fig.~\ref{fig1d}.

\section{Construction of the transformation function $\mathcal{T}(x)$ for the Lennard-Jones system.}\label{tranform}

Since the Lennard-Jones system is invariant to translation of the particles in the system, we fix the coordinate of particle $1$ at the origin, {\it i.e.} $y_1=(0,0)$ in the example. Thus the system is defined in the $12$-dimensional space with the configuration $x=(y_2,\dots,y_7)$, where $y_i$ denotes the coordinate of particle $i$.

Before introducing the transformation function, we define two angle functions as follows.

{\bf Definition.} {\it The angle $f(u,v)\in [0,\pi]$ between two nonzero vectors $u=(u_1,u_2)$ and $v=(v_1,v_2)$ in $\R^2$ is defined as 
\begin{equation*}
    f(u,v)=\arccos\left(\dfrac{\langle u,v\rangle}{|u|\cdot|v|}\right).
\end{equation*}
Then we define the oriented angle between $u$ and $v$ as
\begin{equation*}
    \eta(u,v) = 
\begin{cases}
    f(u,v),& \langle u',v\rangle\geq 0\\
    2\pi-f(u,v),& \langle u',v\rangle<0\\
\end{cases}
\end{equation*}
where the vector $u'=(-u_2,u_1)$. Note that $\eta(u,v)\in[0,2\pi]$.}

Next, we construct a transformation function $\mathcal{T}(x)$ for the system in a two-step procedure.

\begin{figure}
\centering
\includegraphics[width=.8\linewidth]{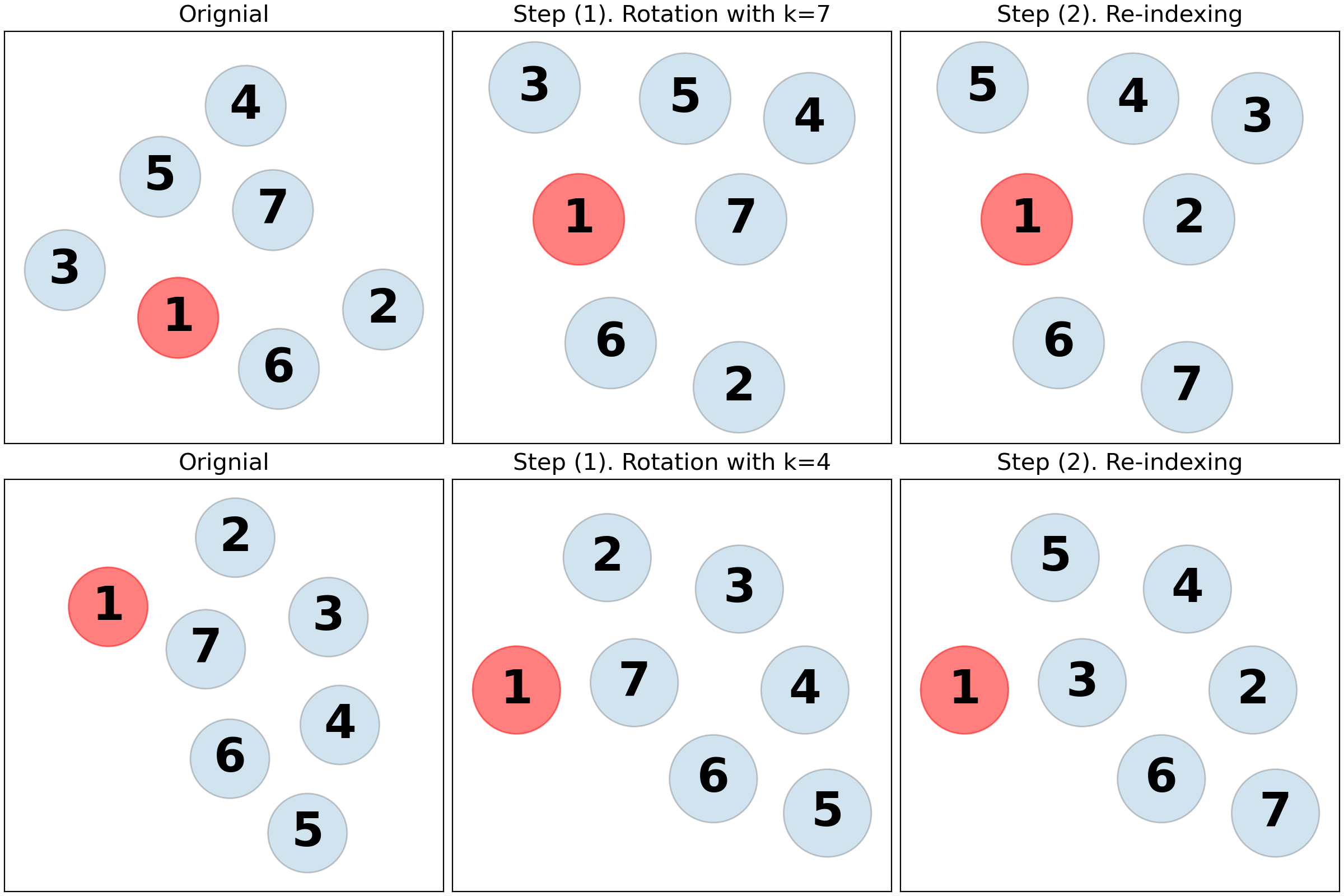}
\caption{Illustration of the  two-step procedure for the transformation function $\mathcal{T}(x)$ using two example clusters of the Lennard-Jones system.}
\label{figA2}
\end{figure}

\subsection*{(1) Rotating the cluster}
We compute the mean vector
\begin{equation*}
    \bar{y} = \frac{1}{6}(y_2+\dots+y_7).
\end{equation*}
Let $k$ be the number that solves
\begin{equation*}
    k = \argmin_{2\leq i\leq 7} f(y_i,\bar{y}).
\end{equation*}

Using the particle $k$, we set the angle 
\begin{equation*}
    \beta = \eta(y_k,e_x),
\end{equation*}
where $e_x$ denotes the unit vector $e_x=(1,0)$, and construct the rotation matrix 
\begin{equation*}
    D_{\beta} = \left[
\begin{matrix}
    \cos\beta & -\sin\beta\\
    \sin\beta & \cos\beta\\
\end{matrix}\right].
\end{equation*}
We rotate the cluster with configuration $x$ to $x'=(y'_2,\dots,y'_7)$,
\begin{equation*}
    y'_i=D_{\beta} \left[
    \begin{matrix}
        y_{i1}\\
        y_{i2}
    \end{matrix}\right],\quad 2\leq i\leq 7
\end{equation*}
where $y_i=(y_{i1},y_{i2})$ denotes the coordinates of particle $i$. In the new configuration $x'$, the particle $k$ is on the $x$-axis.

\subsection*{(2) Re-indexing the particles}
We sort the angles $\{\eta(e_x,y_i')\}$, $2\leq i\leq 7$ of the particles and obtain the particle sequence $( y'_{\tau(2)},\dots,y'_{\tau(7)})$, where $(\tau(2),\dots,\tau(7))$ denote the indices of the sorted particles.

Therefore, the transformation function is defined as: $\mathcal{T}(x)=(y'_{\tau(2)},\dots,y'_{\tau(7)})$. In Fig.~\ref{figA2}, we use two example clusters of the Lennard-Jones system to illustrate the above procedure for $\mathcal{T}(x)$.

To incorporate the transformation function $\mathcal{T}(x)$ into the critic and actor networks. We compute the critic function $Q_{\theta}$ and actor function $\mu_{\theta}$ at the state $s\in\R^{12}$ and action $a$ which is a unit vector in $\R^{12}$ as follows. We compute $s'=\mathcal{T}(s)$ with rotation matrix $D_{\beta}$ and indices $(\tau(2),\dots,\tau(7))$. Then we compute the transformed action $a'$ from $a$ using the same rotation $D_{\beta}$ and index function $\tau$. The critic function at $(s,a)$ is given by $Q_{\theta}(s',a')$.

Denote $\tilde{a}=\mu_{\theta}(s')$ and let $\tau^{-1}$ be the inverse function of $\tau$. We compute the transformed action  $\tilde{a}'$ from $\tilde{a}$ using the rotation $D_{-\beta}$ and index function $\tau^{-1}$.
The actor function at $s$ is given by $\tilde{a}'$.

\end{document}